\pdfoutput=1

\RequirePackage[l2tabu,orthodox]{nag}

\documentclass[11pt,letterpaper,reqno,oneside]{article}

\DeclareTextFontCommand{\emph}{\slshape}

\makeatletter
\renewcommand{\paragraph}{%
	\@startsection{paragraph}{4}%
	{\z@}{1.75ex \@plus 1ex \@minus .2ex}{-0.7em}%
	{\normalfont\normalsize\bfseries}%
}
\makeatother

\usepackage[T1]{fontenc}
\usepackage{lmodern}
\usepackage[utf8]{inputenc}

\usepackage[american,german,british]{babel}

\usepackage[activate={true,nocompatibility}]{microtype}

\usepackage{csquotes}

\usepackage[backend=biber,autolang=other,style=apa,maxcitenames=5,sortcites=false]{biblatex}
\DeclareLanguageMapping{british}{british-apa}
\DeclareLanguageMapping{american}{american-apa}
\usepackage{amsmath}
\usepackage{amssymb}
\usepackage{amsfonts}
\usepackage{amsthm}
\usepackage{mathtools}
\usepackage{stmaryrd}

\let\originalleft\left
\let\originalright\right
\renewcommand{\left}{\mathopen{}\mathclose\bgroup\originalleft}
\renewcommand{\right}{\aftergroup\egroup\originalright}

\usepackage{graphicx}
\usepackage{tikz,pgfplots}
\pgfplotsset{compat=1.10}
\usetikzlibrary{shapes.multipart}
\usetikzlibrary{decorations.markings}

\usepackage[title]{appendix}

\usepackage{datetime}
\newdateformat{datestyle}{ {\THEDAY} \monthname[\THEMONTH] \THEYEAR }

\usepackage{enumerate}
\usepackage{enumitem}
\setlist[enumerate,1]{label=(\arabic*)}
\setlist[itemize,1]{label=--}
\setlist[itemize,2]{label=--}
\setlist[itemize,3]{label=--}
\setlist[itemize,4]{label=--}

\usepackage{xcolor}
\usepackage{verbatim}
\usepackage{gensymb}
\usepackage{rotating}
\usepackage{subcaption}
\usepackage{floatpag}
\floatpagestyle{plain}
\usepackage{marvosym}

\usepackage{hyphenat}
\theoremstyle{definition}

\newtheorem{theorem}{Theorem}
\newtheorem{proposition}{Proposition}
\newtheorem{lemma}{Lemma}
\newtheorem{corollary}{Corollary}
\newtheorem{remark}{Remark}
\newtheorem{observation}{Observation}
\newtheorem{example}{Example}
\newtheorem{definition}{Definition}
\newtheorem*{claim}{Claim}

\newtheoremstyle{named}
	{\topsep}					
	{\topsep}					
	{}							
	{0pt}						
	{\bfseries}					
	{}							
	{5pt plus 1pt minus 1pt}	
	{\thmnote{#3}}				
\theoremstyle{named}
\newtheorem{namedthm}{}

\renewcommand{\qedsymbol}{$\blacksquare$}

\usepackage{xpatch}
\xpatchcmd{\proof}{\itshape}{\proofheadfont}{}{}
\newcommand{\proofheadfont}{\slshape}

\usepackage{hyperref}
\hypersetup{pdfborder=0 0 0}

\usepackage[nameinlink]{cleveref}
\crefname{page}{p.}{pp.}
\crefname{equation}{equation}{equations}
\crefname{section}{section}{sections}
\crefname{subsection}{section}{sections}
\crefname{subsubsection}{section}{sections}
\crefname{appsec}{appendix}{appendices}
\crefname{supplsec}{supplemental appendix}{supplemental appendices}
\crefname{footnote}{footnote}{footnotes}
\crefname{figure}{figure}{figures}
\crefname{table}{table}{tables}
\crefname{theorem}{theorem}{theorems}
\crefname{proposition}{proposition}{propositions}
\crefname{lemma}{lemma}{lemmata}
\crefname{corollary}{corollary}{corollaries}
\crefname{remark}{remark}{remarks}
\crefname{observation}{observation}{observations}
\crefname{example}{example}{examples}
\crefname{fact}{fact}{facts}
\crefname{definition}{definition}{definitions}
\crefname{assumption}{assumption}{assumptions}
\crefname{exercise}{exercise}{exercises}
\crefname{notation}{notation}{notation}
\crefname{claim}{claim}{claims}
\crefname{conjecture}{conjecture}{conjectures}

\newcommand{\eps}{\varepsilon}

\newcommand{\dd}{\mathrm{d}}

\DeclareMathOperator*{\interior}{int}

\DeclareMathOperator*{\cl}{cl}

\DeclareMathOperator*{\co}{co}

\DeclareMathOperator*{\supp}{supp}

\newcommand{\PP}{\mathbf{P}}

\newcommand{\R}{\mathbf{R}}

\newcommand{\N}{\mathbf{N}}

\newcommand{\1}{\boldsymbol{1}}

\newcommand{\union}{\cup}

\newcommand{\intersect}{\cap}

\DeclarePairedDelimiter\abs{\lvert}{\rvert}

\newcommand*{\xslant}[2][76]{%
	\begingroup
	\sbox0{#2}%
	\pgfmathsetlengthmacro\wdslant{\the\wd0 + cos(#1)*\the\wd0}%
	\leavevmode
	\hbox to \wdslant{\hss
		\tikz[
			baseline=(X.base),
			inner sep=0pt,
			transform canvas={xslant=cos(#1)},
		] \node (X) {\usebox0};%
		\hss
		\vrule width 0pt height\ht0 depth\dp0 %
	}%
	\endgroup
}
\makeatletter
\newcommand*{\xslantmath}{}
\def\xslantmath#1#{%
	\@xslantmath{#1}%
}
\newcommand*{\@xslantmath}[2]{%
	\ensuremath{%
		\mathpalette{\@@xslantmath{#1}}{#2}%
	}%
}
\newcommand*{\@@xslantmath}[3]{%
	\xslant#1{$#2#3\m@th$}%
}
\makeatother

\makeatletter
\def\namedlabel#1#2{\begingroup
	#2%
	\def\@currentlabel{#2}%
	\phantomsection\label{#1}\endgroup
}
\makeatother

\makeatletter
\let\save@mathaccent\mathaccent
\newcommand*\if@single[3]{%
	\setbox0\hbox{${\mathaccent"0362{#1}}^H$}%
	\setbox2\hbox{${\mathaccent"0362{\kern0pt#1}}^H$}%
	\ifdim\ht0=\ht2 #3\else #2\fi
	}
\newcommand*\rel@kern[1]{\kern#1\dimexpr\macc@kerna}
\newcommand*\widebar[1]{\@ifnextchar^{{\wide@bar{#1}{0}}}{\wide@bar{#1}{1}}}
\newcommand*\wide@bar[2]{\if@single{#1}{\wide@bar@{#1}{#2}{1}}{\wide@bar@{#1}{#2}{2}}}
\newcommand*\wide@bar@[3]{%
	\begingroup
	\def\mathaccent##1##2{%
	  \let\mathaccent\save@mathaccent
	  \if#32 \let\macc@nucleus\first@char \fi
	  \setbox\z@\hbox{$\macc@style{\macc@nucleus}_{}$}%
	  \setbox\tw@\hbox{$\macc@style{\macc@nucleus}{}_{}$}%
	  \dimen@\wd\tw@
	  \advance\dimen@-\wd\z@
	  \divide\dimen@ 3
	  \@tempdima\wd\tw@
	  \advance\@tempdima-\scriptspace
	  \divide\@tempdima 10
	  \advance\dimen@-\@tempdima
	  \ifdim\dimen@>\z@ \dimen@0pt\fi
	  \rel@kern{0.6}\kern-\dimen@
	  \if#31
	    \overline{\rel@kern{-0.6}\kern\dimen@\macc@nucleus\rel@kern{0.4}\kern\dimen@}%
	    \advance\dimen@0.4\dimexpr\macc@kerna
	    \let\final@kern#2%
	    \ifdim\dimen@<\z@ \let\final@kern1\fi
	    \if\final@kern1 \kern-\dimen@\fi
	  \else
	    \overline{\rel@kern{-0.6}\kern\dimen@#1}%
	  \fi
	}%
	\macc@depth\@ne
	\let\math@bgroup\@empty \let\math@egroup\macc@set@skewchar
	\mathsurround\z@ \frozen@everymath{\mathgroup\macc@group\relax}%
	\macc@set@skewchar\relax
	\let\mathaccentV\macc@nested@a
	\if#31
	  \macc@nested@a\relax111{#1}%
	\else
	  \def\gobble@till@marker##1\endmarker{}%
	  \futurelet\first@char\gobble@till@marker#1\endmarker
	  \ifcat\noexpand\first@char A\else
	    \def\first@char{}%
	  \fi
	  \macc@nested@a\relax111{\first@char}%
	\fi
	\endgroup
}
\makeatother

\DeclareMathOperator*{\cotwo}{\underline{co}}
\DeclareMathOperator*{\inftwo}{\underline{inf}}

\title{\scshape Outside options and risk attitude%
\thanks{We are grateful for comments from Eddie Dekel, Prajit Dutta, Jeff Ely, Alkis Georgiadis-Harris, En Hua Hu, Ravi Jagadeesan, Ian Jewitt, Paul Klemperer, Krittanai Laohakunakorn, Elliot Lipnowski, Sujoy Mukerji, Aidan Smith, Omer Tamuz, Juuso Toikka, Leeat Yariv, and audiences at Helsinki GSE, the 2025 Transatlantic Theory Workshop, the 2025 Southeast Theory Festival, and the 2025 City--QMUL Workshop in Economic Theory. Oscar Calvert and Augustus Smith provided excellent research assistance. Curello acknowledges support from the German Research Foundation (DFG) through CRC TR 224 (Project B02).}}

\author{%
Gregorio Curello \\
Mannheim
\and
Ludvig Sinander \\
Oxford
\and
Mark Whitmeyer \\
Arizona State}

\date{18 September 2025}

\makeatletter
	\AtBeginDocument{ \hypersetup{
		pdftitle = {Outside options and risk attitude},
		pdfauthor = {Gregorio Curello, Ludvig Sinander and Mark Whitmeyer}
		} }
\makeatother

\begin{document}

\maketitle

\begin{abstract}
	We uncover a close link between outside options and risk attitude: when a decision-maker gains access to an outside option, her behaviour becomes less risk-averse, and conversely, any observed decrease of risk-aversion can be explained by an outside option having been made available. We characterise the comparative statics of risk-aversion, delineating how effective risk attitude (i.e. actual choice among risky prospects) varies with the outside option and with the decision-maker's `true' risk attitude. We prove that outside options are special: among transformations of a decision problem, those that amount to adding an outside option are the only ones that always reduce risk-aversion.
\end{abstract}

\section{Introduction}
\label{sec:intro}

Many important economic decisions amount to choosing between risky prospects, usually modelled as lotteries. Examples include project choice in firms, household financial decisions such as selecting an investment portfolio or choosing between health-insurance plans, and life-cycle decisions such as choosing one's occupation, spouse and location. A decision-maker's \emph{(effective) risk attitude} summarises how she chooses in such situations.

In most economic contexts, the chosen prospect is not all there is. In project choice, there is (usually) limited liability: if the chosen project fails spectacularly, then the firm can declare bankruptcy, and responsible managers can leave. Households are likewise protected by limited liability (personal bankruptcy), and can fall back on the public option when faced with an uninsured health shock. An individual dissatisfied with her occupation or marriage can retrain or divorce, and if she chose a location prone to natural disaster (such as Florida) then she may expect a bailout (e.g. from FEMA). All of these are examples of \emph{outside options,} which can be exercised ex post, after the outcome of the risky prospect becomes known.

In this paper, we explore how outside options shape effective risk attitude. In our model, a decision-maker is characterised by her `true' risk attitude, assumed to be expected-utility, and by her outside option. The outside option can be random, and may be unavailable with some probability. The decision-maker decides ex post whether to exercise the outside option. Thus when considering any given alternative (`inside option'), she recognises that she will consume that alternative only if the outside option turns out to be worse. Our question is how this affects her choices among risky prospects.

Our inquiry parallels the classic question of how background risk shapes effective risk attitude, as in e.g. \textcite{PomattoStrackTamuz2020,MuPomattoStrackTamuz2024}.%
	\footnote{`Background risk' means randomness other than that embodied in the risky prospects themselves; for example, a household choosing among portfolios may also face income risk. The classics of this literature include \textcite{Ross1981,KihlstromRomerWilliams1981,PrattZeckhauser1987,Kimball1993,EeckhoudtGollierSchlesinger1996,GollierPratt1996}.}
A key take-away from this literature is that `it depends': whether one risky prospect remains preferred to another after background risk is introduced depends on the details of the two prospects and the distribution of the background risk.

We show, by contrast, that the impact of outside options on risk attitude is unambiguous and clean. In particular, our main theorem states that granting a decision-maker access to an outside option always makes her less risk-averse (in the standard \textcite{Yaari1969} sense), and conversely that any decrease of risk-aversion can be rationalised as arising from an outside option having been made available. This can be viewed as an identification result: by observing a decision-maker's choices (her effective risk attitude), we obtain a non-parametric lower bound on her `true' risk-aversion, and can learn nothing further about her `true' risk attitude.

Our second theorem delves deeper, characterising how effective risk attitude varies with the `true' risk attitude and the outside option. When the outside option is held fixed and the `true' risk attitude shifts, we show that the effective risk attitude becomes less risk-averse if and only if the `true' risk attitude does. When the `true' risk attitude is held fixed and the outside option shifts, the effective risk attitude becomes less risk-averse if and only if the distribution of the value of the outside option improves in the reverse hazard rate order, a (standard) notion that is intermediate between first-order stochastic dominance and the likelihood-ratio order.

Finally, we argue that the tight link identified in our main theorem between reducing risk-aversion and adding an outside option is special. For example, introducing background risk does not always reduce risk-aversion.%
	\footnote{\label{footnote:special_background}This is known: see \textcite{EeckhoudtGollierSchlesinger1996,GollierPratt1996}.}
We prove a general theorem to this effect: among transformations of a decision problem, those that amount to adding an outside option are the \emph{only} ones that always reduce risk-aversion.

\subsection{Related literature}
\label{sec:intro:lit}

One half of our main theorem formalises the intuitive idea that having access to a fallback incentivises greater risk-taking. This is an old idea: for example, Adam Smith (\citeyear{Smith1776}) argued that limited liability has this effect.%
	\footnote{`In a private copartnery, each partner is bound for the debts contracted by the company to the whole extent of his fortune. In a joint stock company, on the contrary, each partner is bound only to the extent of his share. This total exemption from trouble and from risk, beyond a limited sum, encourages many people to become adventurers in joint stock companies, who would, upon no account, hazard their fortunes in any private copartnery.' (\cite{Smith1976}, pp.~740--1.)}
(See also e.g. \textcite[][section~4.1]{JensenMeckling1976}, \textcite{Golbe1981,Golbe1988,GollierKoehlRochet1997}.) This idea has important practical implications, e.g. for financial stability.%
	\footnote{Several financial crises have arisen in part from bankers placing risky bets in the knowledge that they will not be held personally liable if large losses result \parencite[see e.g.][]{FCIC2011}.}

The other half of our main theorem, asserting that any decrease of risk-aversion can be generated by an outside option having been made available, has no close parallels to our knowledge. The same applies to our second and third theorems.

We view our analysis as contributing to the broader agenda of understanding how \emph{economic} forces influence decision-makers' effective risk attitudes. Other literatures in this vein consider the effect of background risk,%
	\footnote{E.g. \textcite{Ross1981,KihlstromRomerWilliams1981,PrattZeckhauser1987,Kimball1993,EeckhoudtGollierSchlesinger1996,GollierPratt1996,PomattoStrackTamuz2020,MuPomattoStrackTamuz2024}.}
contracts,%
	\footnote{Including employment contracts \parencite[e.g.][]{Wilson1969,Ross1974,AmihudLev1981,Lambert1986,HirshleiferSuh1992,Diamond1998,GaricanoRayo2016,BarronGeorgiadisSwinkels2020} and financing contracts \parencite[e.g.][]{GalaiMasulis1976,JensenMeckling1976,StiglitzWeiss1981,Green1984,Hebert2018}.}
having an audience,%
	\footnote{For example, career concerns (e.g. Holmström, \citeyear{Holmstrom1982}/\citeyear{Holmstrom1999}, \cite{HirshleiferThakor1992,Hermalin1993,Chen2015}) or disclosure \parencite{BenporathDekelLipman2018}.}
(in)flexibility,%
	\footnote{E.g. Drèze and Modigliani (\citeyear{DrezeModigliani1966}/\citeyear{DrezeModigliani1972}), \textcite{Mossin1969,SpenceZeckhauser1972,Machina1982,Machina1984,BodieMertonSamuelson1992,Gollier2005,ChettySzeidl2007,PostlewaiteSamuelsonSilverman2008}.}
and competition.%
	\footnote{For example, R\&D races \parencite[e.g.][]{DasguptaStiglitz1980,KletteDemeza1986,DasguptaMaskin1987,Cabral2003,AndersonCabral2007}, competition for status \parencite[e.g.][]{Robson1992,Rosen1997,BeckerMurphyWerning2005,RayRobson2012,Hopkins2018}, contests more generally \parencite[e.g.][]{Hvide2002,HvideKristiansen2003,Taylor2003,KrakelSliwka2004,SeelStrack2013,SeelStrack2016,NutzZhang2022,FangKeKubitzNoePageLiu2025}, and competition between banks \parencite[see the surveys by][]{Beck2008,Carletti2008,Vives2016,BergerKlapperTurkariss2017}.}
A somewhat analogous literature studies how natural selection shapes risk attitude.%
	\footnote{E.g. \textcite{Robson1996jet,Robson1996geb,Robson2001jpe,DekelScotchmer1999,Netzer2009,RobattoSzentes2017,Georgiadisharris2017}.}

An alternative way of viewing our main result is as a new characterisation (in terms of outside options) of `less risk-averse than' for expected-utility preferences. From this perspective, the theorem continues the enterprise of Pratt's (\citeyear{Pratt1964}) theorem, much like e.g. \textcite[][Theorem~4]{MaccheroniMarinacciWangWu2025}.

\subsection{Roadmap}
\label{sec:intro:roadmap}

We outline the setup and relevant background in the next section, then describe the outside-option model in \cref{sec:oo}. In \cref{sec:fixedF}, we prove a basic result (\Cref{proposition:fixedF}) that describes how outside options shape risk attitude. We then characterise the outside-option model's identification properties (\cref{sec:ident}, particularly \Cref{theorem:ident}) and comparative-statics properties (\cref{sec:mcs}, \Cref{theorem:mcs}). Finally (\cref{sec:special}, \Cref{theorem:special}), we show that transformations other than adding an outside option do not generally reduce risk-aversion. All proofs omitted from the main text may be found in the appendix.

\section{Setup and background}
\label{sec:setup}

There is a non-empty set $X$ of alternatives, with generic elements $x,y,z,w \in X$. We consider simple lotteries, meaning (probability mass functions of) finitely supported probability distributions over $X$. Formally, a \emph{simple lottery} is a function $p : X \to [0,1]$ such that $\supp(p) \coloneqq \{ x \in X : p(x) > 0 \}$ is finite and $\sum_{x \in \supp(p)} p(x) = 1$. We write $\Delta^0(X)$ for the set of all simple lotteries, with generic elements $p,q,r \in \Delta^0(X)$. By the usual abuse, the lottery in $\Delta^0(X)$ that is degenerate at $x \in X$ is denoted simply `$x$'.

A \emph{preference} is a complete and transitive binary relation on $\Delta^0(X)$. A decision-maker's preference is, at least in principle, an empirical object: it can be recovered from (sufficiently rich) choice data.

A preference $\succeq$ is called \emph{expected-utility} if and only if there exists a function $u : X \to \R$ such that for any simple lotteries $p,q \in \Delta^0(X)$, $p \succeq q$ if and only if $\int u \dd p \geq \int u \dd q$. (Here `$\int u \dd p$' is standard shorthand for $\sum_{x \in \supp(p)} p(x) u(x)$.) The function $u : X \to \R$ is called a \emph{risk attitude} (or `vNM utility function'), and is said to \emph{represent} $\succeq$. Recall:

\begin{namedthm}[von Neumann--Morgenstern (\citeyear{VonneumannMorgenstern1947}) theorem.]
	\label{theorem:vNM}
	A preference is expected-utility if and only if it is continuous and satisfies the independence axiom. Two risk attitudes $u,v : X \to \R$ represent the same expected-utility preference if and only if there exist $\alpha>0$ and $\beta \in \R$ such that $\alpha u + \beta = v$.
\end{namedthm}

We shall make extensive use of Lipschitz continuity, absolute continuity, and the (Lebesgue) fundamental theorem of calculus.%
	\footnote{See e.g. \textcite[][pp.~105, 106 and 108]{Folland1999}.}
Our notation is standard, except for two symbols `$\cotwo$' and `$\inftwo$' that we define below.

\begin{namedthm}[Standard notation.]
	\label{namedthm:notation}
	For any set $A \subseteq \R$, we write $\co(A)$ for its convex hull, $\cl(A)$ for its closure (in $\R$), and $\interior(A)$ for its interior (in $\R$). We write $[-\infty,+\infty) \coloneqq \R \cup \{-\infty\}$, $(-\infty,+\infty] \coloneqq \R \cup \{+\infty\}$ and $[-\infty,+\infty] \coloneqq \R \cup \{-\infty,+\infty\}$, and for each $a \in [-\infty,+\infty]$, $(a,a] \coloneqq \varnothing \eqqcolon [a,a)$ and $(a,a) \coloneqq (a,a] \cap [a,a) = \varnothing$. For non-empty sets $A$, $B$ and $C$ and functions $f : A \to B$ and $g : B \to C$, the composition of $f$ and $g$ is denoted $g \circ f$, the image of $f$ is denoted $f(A) \subseteq B$, and if $f$ is injective, then its inverse (a map $f(A) \to A$) is denoted $f^{-1}$. For any sets $A \subseteq B \neq \varnothing$, $\1_A : B \to \{0,1\}$ denotes the function defined by $\1_A(a) \coloneqq 0$ if $a \notin A$ and $\1_A(a) \coloneqq 1$ if $a \in A$, and `the identity (on $A$)' means the function $f : A \to A$ given by $f(a) \coloneqq a$ for each $a \in A$.
\end{namedthm}

\begin{definition}
	\label{definition:cotwo}
	For any set $A \subseteq \R$, define $\cotwo(A) \subseteq A$ and $\inftwo A \in \cl(A)$ by
	\begin{align*}
		&\begin{aligned}
			\cotwo(A) &\coloneqq \co( A \setminus \{\inf A\} ) \cup \{\inf A\}
			\\
			\inftwo A &\coloneqq \inf ( A \setminus \{\inf A\} )
		\end{aligned}
		&&\Biggr\} \quad \text{if $\inf A < \inf( A \setminus \{\inf A\} ) \notin A$}
		\\
		&\begin{aligned}
			\cotwo(A) &\coloneqq \co(A)
			\\
			\inftwo A &\coloneqq \inf A
		\end{aligned}
		&&\Biggr\} \quad \text{otherwise.}
	\end{align*}
\end{definition}

Evidently $A \subseteq \cotwo(A) \subseteq \co(A)$ and $\co(A) \setminus \cotwo(A) = (\inf A,\inftwo A]$. If $A$ is convex, then $A = \cotwo(A) = \co(A)$.

\subsection{Comparative risk attitude}
\label{sec:setup:lra}

Recall Yaari's (\citeyear{Yaari1969}) definition of comparative risk-aversion: for any two preferences $\succeq$ and $\succeq'$, $\succeq$ is called \emph{less risk-averse than} $\succeq'$ if and only if for each alternative $x \in X$ and each simple lottery $p \in \Delta^0(X)$, $x \succeq \mathrel{(\succ)} p$ implies $x \succeq' \mathrel{(\succ')} p$. Further recall Pratt's (\citeyear{Pratt1964}) theorem, which characterises `less risk-averse than' for expected-utility preferences:

\begin{namedthm}[Pratt's theorem (part~1).]
	\label{theorem:pratt}
	For a non-empty set $X$ and functions $u,v : X \to \R$, the following are equivalent:

	\begin{enumerate}[label=(\Alph*)]
	
		\item \label{item:pratt:lra} $u$ is less risk-averse than $v$, i.e. for any alternative $x \in X$ and simple lottery $p \in \Delta^0(X)$, $u(x) \geq \mathrel{(>)} \int u \dd p$ implies $v(x) \geq \mathrel{(>)} \int v \dd p$.

		\item \label{item:pratt:trans} There exists an increasing convex function $\phi : \co(v(X)) \to \R$ that is strictly increasing on $v(X)$ and satisfies $u = \phi \circ v$.

		\item \label{item:pratt:curv} The following two properties hold:
		\begin{enumerate}[label=(\Roman*),topsep=0em]
		
			\item \label{item:pratt:curv:ordequiv} For any $x,y \in X$, $u(x) \geq \mathrel{(>)} u(y)$ implies $v(x) \geq \mathrel{(>)} v(y)$.

			\item \label{item:pratt:curv:compress} For any alternatives $x,y,z \in X$, if $u(x) < u(y) < u(z)$, then
			\begin{equation*}
				\frac{u(z)-u(y)}{u(y)-u(x)}
				\geq \frac{v(z)-v(y)}{v(y)-v(x)} .
			\end{equation*}
		
		\end{enumerate}
	
	\end{enumerate}
\end{namedthm}

\begin{remark}
	\label{remark:generalX}
	In the literature, comparative risk-aversion and its characterisation are almost only ever considered in the special case in which alternatives are monetary prizes: $X \subseteq \R$. However, by inspection, properties \ref{item:pratt:lra}--\ref{item:pratt:curv} are equally meaningful whatever the nature of the alternatives, and in fact (as asserted above) they are equivalent even outside of the monetary-prizes case. This observation is not (entirely) new,%
		\footnote{For example, \textcite{WakkerPetersVanriel1986} make essentially this point, and prove the equivalence of a variant of \ref{item:pratt:lra} with a variant of \ref{item:pratt:trans} (their Theorem~3.1).}
	but it is not widely known, and we have not found a complete proof in the literature. We therefore include a proof of \hyperref[theorem:pratt]{Pratt's theorem (part~1)} in \cref{app:pratt_pf}.
\end{remark}

Property~\ref{item:pratt:trans} asserts strict monotonicity of $\phi$ only on $v(X)$, but $\phi$ can in fact be chosen to be strictly increasing on $\cotwo(v(X))$.%
	\footnote{Recall that `$\cotwo$' was defined on \cpageref{definition:cotwo} and satisfies $v(X) \subseteq \cotwo(v(X)) \subseteq \co(v(X))$.}
It need not be possible to choose $\phi$ to be strictly increasing on its full domain $\co(v(X))$, however, nor need it be possible to choose $\phi$ to be continuous. These claims are (fully stated and) proved in \cref{app:pratt_phis}.

Further recall that if the alternatives are monetary prizes (i.e. $X \subseteq \R$ with strictly increasing utility $u : X \to \R$), then `less risk-averse than' is characterised by pointwise inequality of Arrow--Pratt indices:

\begin{namedthm}[Pratt's theorem (part~2).]
	\label{theorem:pratt_diff}
	For a non-empty open convex subset $X$ of $\R$ and twice continuously differentiable functions $u,v : X \to \R$ satisfying $u' > 0 < v'$, \ref{item:pratt:lra}--\ref{item:pratt:curv} hold if and only if $u''/u' \geq v''/v'$.
\end{namedthm}

For completeness, we include a proof in \cref{app:pratt_diff_pf}. A third part of Pratt's theorem, which we will not use, asserts that for a non-empty convex subset $X$ of $\R$ and continuous strictly increasing $u,v : X \to \R$, \ref{item:pratt:lra}--\ref{item:pratt:curv} hold if and only if for every simple lottery $p \in \Delta^0(X)$, $u^{-1}\left( \int u \dd p \right) \geq v^{-1}\left( \int v \dd p \right)$.%
	\footnote{The equivalence of this condition with \ref{item:pratt:trans} was shown already by \textcite[][Theorem~92]{HardyLittlewoodPolya1934}, though without the expected-utility interpretation.}

\subsection{Extended CDFs}
\label{sec:setup:ecdf}

We shall work with random variables taking values in $[-\infty,+\infty) \coloneqq \R \cup \{-\infty\}$. Distributions of $[-\infty,+\infty)$-valued random variables may be described by their cumulative distribution functions, as follows.

\begin{definition}
	\label{definition:ecdf}
	A function $F : \R \to [0,1]$ is called an \emph{extended CDF} if and only if there exists a $[-\infty,+\infty)$-valued random variable $K$ such that $F(k) = \PP( K \leq k )$ for every $k \in \R$.
\end{definition}

\begin{observation}
	\label{observation:ecdf}
	A function $F : \R \to [0,1]$ is an extended CDF if and only if it is increasing and right-continuous and satisfies $\lim_{k \nearrow +\infty} F(k)=1$.%
		\footnote{The `only if' part is trivial, and the `if' part holds since for any $F : [0,1] \to \R$ with these properties, letting $U$ be a random variable that is distributed uniformly on $[0,1)$, we may consider the $[-\infty,+\infty)$-valued random variable $K$ defined by $K \coloneqq \inf\{ k \in \R : F(k) \geq U \}$.}
\end{observation}

For an extended CDF $F$, the quantity $\alpha \coloneqq \lim_{k \searrow -\infty} F(k)$ is the probability of $-\infty$ (so $F$ is a CDF if and only if $\alpha=0$).

The expectation $\int g \dd F \equiv \int g(k) F(\dd k)$ of a bounded-below Lebesgue-measurable function $g : [-\infty,+\infty) \to \R$ with respect to an extended CDF $F$ is defined in the natural way, namely
\begin{equation*}
	\int g \dd F \coloneqq \left( \lim_{k \searrow -\infty} F(k) \right) g(-\infty) + \int_\R g \dd F ,
\end{equation*}
where the latter integral is understood in the usual Lebesgue--Stieltjes sense.%
	\footnote{If $\alpha \coloneqq \lim_{k \searrow -\infty} F(k)<1$, then the latter term equals $(1-\alpha) \times \frac{1}{1-\alpha} \int_\R g \dd F$: the probability of a finite draw times the expectation of $g$ conditional on a finite draw.}
Evidently the integral $\int g \dd F$ is finite if and only if $\int_\R g \dd F$ is.

Given a set $S \subseteq [-\infty,+\infty)$ such that $S \cap \R$ is Borel, we say that an extended CDF $F$ is \emph{concentrated} on $S$ if and only if $\int_S \dd F = 1$.

\section{The outside-option model}
\label{sec:oo}

We consider a decision-maker who has risk attitude $v : X \to \R$ and who has access to an outside option whose value (in $v$ units) is drawn from an extended CDF $F$. The possibility that the outside option has `value $-\infty$' captures the possibility that the outside option may be unavailable.

The decision-maker's valuation of any given alternative $x \in X$ is then
\begin{equation*}
	\int \max\{v(x),k\} F(\dd k)
	= F(v(x)) v(x)
	+ \int_{(v(x),+\infty)} k F(\dd k ) ,
\end{equation*}
reflecting the idea that the decision-maker will observe the realisation $k$ of (the value of) the outside option before choosing whether to exercise it, and thus exercises it if and only if $k$ exceeds $v(x)$. (The above integrals are finite if and only if $\int_{(0,+\infty)} k F(\dd k) < +\infty$.)

The extended CDF $F$ is a reduced-form way of capturing the implications for utility/choice of an actual physical outside option being drawn from a set $Y \cup \{\varnothing\}$, where `$\varnothing$' means `no outside option available'. Formally, $F$ arises as
\begin{equation}
	F(k) \coloneqq \mu\left( \left\{ y \in Y \cup \{\varnothing\} : w(y) \leq k \right\} \right) \quad \text{for every $k \in \R$,}
	\label{eq:F_from_mu}
\end{equation}
where $\mu$ is a probability measure on (some $\sigma$-algebra on) $Y \cup \{\varnothing\}$, $w : Y \to \R$ specifies the payoff (in $v$ units) of each $y \in Y$, and $w(\varnothing) \coloneqq -\infty$.

\begin{remark}
	\label{remark:oo_X}
	In some applications, the set of possible outside options is $Y=X$ (so $w=v$). In this case, provided $v(X)$ is Borel, \eqref{eq:F_from_mu} implies that $F$ is concentrated on $v(X) \cup \{-\infty\}$. Conversely, if $v(X)$ is Borel, then any extended CDF $F$ that is concentrated on $v(X) \cup \{-\infty\}$ arises via \eqref{eq:F_from_mu} from some probability measure $\mu$ on (some $\sigma$-algebra on) $X \cup \{\varnothing\}$.%
		\footnote{Assume that $v$ is one-to-one; this is without loss since it amounts to replacing $X$ by its quotient with respect to payoff equivalence. Topologise $X$ by embedding it in $\R$ via $x \mapsto v(x)$. Give $X$ the Borel $\sigma$-algebra $\mathcal{B}$, let $\Sigma$ be the $\sigma$-algebra on $X \cup \{\varnothing\}$ generated by $\mathcal{B} \union \{\{\varnothing\}\}$, and define $\mu : \Sigma \to [0,1]$ by $\mu(S) \coloneqq \int_{v(S)} \dd F$ for each $S \in \Sigma$.}
\end{remark}

\subsection{Formal definition and relation to expected utility}
\label{sec:oo:decision}

In terms of behaviour (i.e. choice among / preferences over lotteries), the outside-option model is formally defined as follows.

\begin{definition}
	\label{definition:oo}
	Given a preference $\succeq$, a function $v : X \to \R$ and an extended CDF $F$, we say that $(v,F)$ is an \emph{outside-option representation} of $\succeq$ if and only if $\int_{(0,+\infty)} k F(\dd k) < +\infty$ and for all simple lotteries $p,q \in \Delta^0(X)$, $p \succeq q$ if and only if
	\begin{equation*}
		\int \left( \int \max\{v(x),k\} F(\dd k) \right) p(\dd x)
		\geq \int \left( \int \max\{v(x),k\} F(\dd k) \right) q(\dd x) .
	\end{equation*}
\end{definition}

Note that the inside option (drawn from a lottery $p$) is assumed to be statistically independent of the outside option (drawn from $F$).

The empirical content of the outside-option model is the same as that of expected utility:

\begin{observation}
	\label{observation:risk_att}
	Given a preference $\succeq$, a function $v : X \to \R$ and an extended CDF $F$, if $(v,F)$ is an outside-option representation of $\succeq$, then for any $\alpha>0$ and $\beta \in \R$, $\succeq$ is expected-utility with risk attitude $u : X \to \R$ satisfying
	\begin{equation*}
		\alpha u(x) + \beta = \int \max\{v(x),k\} F(\dd k)
		\quad \text{for every $x \in X$.}
	\end{equation*}
\end{observation}

The function $u$ is the decision-maker's \emph{effective} risk attitude, describing how she actually chooses among lotteries. By the \hyperref[theorem:vNM]{von Neumann--Morgenstern theorem}, the risk attitude $u$ is recoverable (up to positive affine transformations $\ell \mapsto \alpha \ell + \beta$) from the choice data $\succeq$. This contrasts with the parameter $v$, which we interpret as the decision-maker's `true' risk attitude.

\begin{corollary}
	\label{corollary:axiom}
	A preference $\succeq$ admits an outside-option representation if and only if it is expected-utility.
\end{corollary}

\begin{proof}
	If $\succeq$ admits an outside-option representation, then it is expected-utility by \Cref{observation:risk_att}. Conversely, if $\succeq$ is expected-utility with risk attitude $u : X \to \R$, then $(u,1)$ is an outside-option representation of $\succeq$.
\end{proof}

\section{How an outside option shapes risk attitude}
\label{sec:fixedF}

In this section, we give a simple but powerful characterisation of how the presence of an outside option influences risk attitude. Our main results in the subsequent sections are consequences of this fundamental proposition.

To be precise, we seek to answer the following question: given a preference with outside-option representation $(v,F)$ and effective risk attitude $u$ (as per \Cref{observation:risk_att}), how do the properties of the outside-option distribution $F$ shape how the effective risk attitude $u$ differs from the `true' risk attitude $v$? The following result provides the answer.

\begin{proposition}
	\label{proposition:fixedF}
	For a non-empty set $X$, functions $u,v : X \to \R$, and an extended CDF $F$, the following are equivalent:

	\begin{enumerate}[label=(\roman*)]
	
		\item \label{item:fixedF:rep} There exist $\alpha>0$ and $\beta \in \R$ such that
		\begin{equation*}
			\alpha u(x) + \beta
			= \int \max\{v(x),k\} F(\dd k) 
			\quad \text{for every $x \in X$.}
		\end{equation*}

		\item \label{item:fixedF:trans} $\int_{(0,+\infty)} k F(\dd k) < +\infty$, and there exists a $\lambda>0$ and an absolutely continuous $\phi : \co(v(X)) \to \R$ such that $\phi' = \lambda F$ a.e. and $u = \phi \circ v$.
	
	\end{enumerate}
\end{proposition}

Before interpreting property~\ref{item:fixedF:trans}, we outline why \Cref{proposition:fixedF} is true.

\begin{proof}[Sketch proof]
	Note that for (almost) every $\ell \in \co(v(X))$,
	\begin{align*}
		\left. \frac{\dd}{\dd m} \int \max\{m,k\} F(\dd k) \right|_{m=\ell}
		&= \int \left. \frac{\dd}{\dd m} \max\{m,k\} \right|_{m=\ell} F(\dd k)
		\\
		&= \int \1_{[-\infty,\ell]} \dd F
		= F(\ell) .
	\end{align*}
	Evidently \ref{item:fixedF:rep} holds if and only if there exist $\alpha>0$ and $\beta \in \R$ such that $u = \psi \circ v$ holds for $\psi : \co(v(X)) \to \R$ satisfying
	\begin{equation}
		\alpha \psi(\ell) + \beta
		= \int \max\{\ell,k\} F(\dd k)
		\quad \text{for every $\ell \in \co(v(X))$.}
		\label{eq:fixedF_sketch}
	\end{equation}
	Thus if \ref{item:fixedF:rep} holds, then letting $\lambda \coloneqq 1/\alpha$ and $\phi \coloneqq \psi$, we have $\phi' = \lambda F$ a.e. and $u = \phi \circ v$. Conversely, if \ref{item:fixedF:trans} holds, then let $\alpha \coloneqq 1/\lambda$, and note that since the maps $\alpha \phi$ and $\ell \mapsto \int \max\{\ell,k\} F(\dd k)$ have (a.e.) the same derivative, by the fundamental theorem of calculus, they must differ by a constant: in other words, there exists a $\beta \in \R$ such that \eqref{eq:fixedF_sketch} holds with $\psi = \phi$.
\end{proof}

The actual proof (\cref{app:fixedF_pf}) formalises the above argument by taking care of absolute continuity and finiteness of integrals (both of which appear explicitly in property~\ref{item:fixedF:trans}) and by justifying the implicit interchanging in the first step of differentiation and integration.

Property~\ref{item:fixedF:trans} in \Cref{proposition:fixedF} says that the transformation $\phi = u \circ v^{-1}$ of risk attitude from $v$ to $u$ must have gradient proportional to the extended CDF $F$. In case $F$ is continuously differentiable with $F>0$ on $(\inf v(X),+\infty)$, property~\ref{item:fixedF:trans} is equivalent to

\begin{enumerate}[label=(\roman*$'$)]

	\setcounter{enumi}{1}

	\item \label{item:fixedF:trans_diff} $\int_{(0,+\infty)} k F(\dd k) < +\infty$, and there exists a twice continuously differentiable $\phi : \co(v(X)) \to \R$ such that $\phi''/\phi' = F'/F$ on $\co(v(X)) \setminus \{\inf v(X)\}$.%
		\footnote{\ref{item:fixedF:trans} and \ref{item:fixedF:trans_diff} are equivalent because under the hypotheses, $\phi' = \lambda F$ a.e. for some $\lambda > 0$ if and only if $\phi' = \lambda F$ for some $\lambda > 0$ if and only if $\log \phi' = \Lambda + \log F$ for some $\Lambda \in \R$ if and only if $\phi''/\phi' = F'/F$, where the last `if' holds by the fundamental theorem of calculus.}

\end{enumerate}
By \hyperref[theorem:pratt]{Pratt's theorem}, $\phi''/\phi'$ is a measure of how much `more convex' $u$ is than $v$; in other words, it captures how and to what extent risk-aversion is decreased by the addition of the outside option $F$. Property~\ref{item:fixedF:trans_diff} demands that the risk-aversion decrease be equal to the reverse hazard rate $F'/F$ of the outside-option distribution.

\subsection{Monetary alternatives}
\label{sec:fixedF:monetary}

Consider the important special case in which alternatives are monetary prizes: $X$ is an open convex subset of $\R$, the `true' risk attitude $v$ is strictly increasing, and the outside option (if available) is also a monetary prize, so that the \emph{monetary} value of the outside option is a $[-\infty,+\infty)$-valued random variable with extended CDF $G$ satisfying $G = F \circ v$. Then in case $v$ is twice continuously differentiable with $v'>0$ and $G$ is continuously differentiable with $G>0$ on $(\inf X,+\infty)$, \ref{item:fixedF:trans} and \ref{item:fixedF:trans_diff} are equivalent to

\begin{enumerate}[label=(\roman*$''$)]

	\setcounter{enumi}{1}

	\item \label{item:fixedF:trans_diff_money} $\int_{(0,+\infty)} v \dd G < +\infty$, and $u$ is twice continuously differentiable with $u'>0$ and $u''/u' = v''/v' + G'/G$.%
		\footnote{$G' = (F' \circ v) \times v'$, $u' = (\phi' \circ v) \times v'$ and $u'' = (\phi' \circ v) \times v'' + (\phi'' \circ v) \times (v')^2$, so $G'/G = ([F'/F] \circ v) \times v'$ and $u''/u' = v''/v' + ([\phi''/\phi'] \circ v) \times v'$, and thus $\phi''/\phi' = F'/F$ holds if and only if $u''/u' = v''/v' + G'/G$.}

\end{enumerate}

Recall from \hyperref[theorem:pratt_diff]{Pratt's theorem} that $u''/u'$ and $v''/v'$ are measures of risk attitude. The simple additive decomposition of effective risk attitude $u''/u'$ into `true' risk attitude $v''/v'$ and the reverse hazard rate $G'/G$ of the dollar-denominated outside-option distribution has comparative-statics implications, explored in \cref{sec:mcs} below. It also permits parametric analysis:

\begin{example}
	\label{example:cara}
	Suppose that $X = (-\infty,x_0)$ where $x_0 \in \R$, let $u,v : X \to \R$ be CARA, i.e. $u' > 0 < v'$, $-u''/u' = \rho \in \R$ and $-v''/v' = \sigma \in \R$, and let $G$ be a reversed exponential distribution $x \mapsto \min\left\{ 1, \exp\left( - \lambda \left( x_0 - x \right) \right) \right\}$, with reverse hazard rate $\lambda \in [0,+\infty)$. Then \ref{item:fixedF:trans_diff_money} holds if and only if $-\rho = -\sigma + \lambda$.
\end{example}

\section{Identification}
\label{sec:ident}

In this section, we address our central question: what can be learned about the parameters $(v,F)$ from observing choice among lotteries, i.e. from observing the preference $\succeq$? The following simple result asserts that it is equivalent to ask what can be learned about $(v,F)$ from observing the preference $\succeq$'s risk attitude $u : X \to \R$.

\begin{corollary}
	\label{corollary:oo_vF_u}
	Fix $u,v : X \to \R$ and an extended CDF $F$ such that $\int_{(0,+\infty)} k F(\dd k) < +\infty$, let $\succeq$ be an expected-utility preference with risk attitude $u$, and let $\succeq'$ be a preference with outside-option representation $(v,F)$. Then $\mathord{\succeq} = \mathord{\succeq'}$ if and only if there exist $\alpha>0$ and $\beta \in \R$ such that
	\begin{equation*}
		\alpha u(x) + \beta
		= \int \max\{v(x),k\} F(\dd k)
		\quad \text{for every $x \in X$.}
	\end{equation*}
\end{corollary}

\begin{proof}
	By \Cref{observation:risk_att}, $\succeq'$ is expected-utility with risk attitude $w : X \to \R$ given by
	\begin{equation*}
		w(x)
		\coloneqq \int \max\{v(x),k\} F(\dd k)
		\quad \text{for every $x \in X$.}
	\end{equation*}
	By the \hyperref[theorem:vNM]{von Neumann--Morgenstern theorem}, $\mathord{\succeq}=\mathord{\succeq'}$ if and only if there exist $\alpha>0$ and $\beta \in \R$ such that $\alpha u + \beta = w$.
\end{proof}

To build intuition about identification, we begin with an example.

\begin{namedthm}[\Cref*{example:cara} {\normalfont(continued).}]
	\label{example:cara_ident}
	Recall that $-\rho = -\sigma + \lambda$, where $-u''/u' = \rho \in \R$ and $-v''/v' = \sigma \in \R$ are the risk-aversion parameters of, respectively, the effective and `true' risk attitudes, and $G'/G = \lambda \in [0,+\infty)$ is the reverse hazard rate of the dollar-denominated outside-option distribution. Knowledge of the effective risk-aversion parameter $\rho$ provides a lower bound on the `true' risk-aversion parameter: $\sigma \geq \rho$, since $-\rho = -\sigma + \lambda$ and $\lambda \in [0,+\infty)$. There is no upper bound, however, since for any value $\sigma \in [\rho,+\infty)$, $-\rho = -\sigma + \lambda$ holds if $\lambda = \sigma-\rho$. Nothing at all can be learned from $\rho$ about the reverse hazard rate: clearly for any value $\lambda \in [0,+\infty)$, $-\rho = -\sigma + \lambda$ holds if $\sigma = \rho+\lambda$.
\end{namedthm}

In the example, what the effective risk attitude $u$ reveals about the `true' risk attitude $v$ is (precisely) that $v$ is more risk-averse than $u$, while nothing at all is revealed about the outside-option distribution. Our identification results in the next two sections show that these are general lessons.

\subsection{Revealed risk attitude}
\label{sec:ident:riskatt}

The following answers the identification question for $v$: what the effective risk attitude $u$ can teach us about the `true' risk attitude $v$ is nothing more or less than that the latter is more risk-averse than the former.

\begin{theorem}
	\label{theorem:ident}
	For a non-empty set $X$ and bounded-above functions $u,v : X \to \R$ such that $u$ is Lipschitz with respect to $v$,%
		\footnote{That is, there exists an $L \in [0,+\infty)$ such that for any $x,y \in X$, $\abs*{ u(x) - u(y) } \leq L \abs*{ v(x) - v(y) }$.}
	the following are equivalent:

	\begin{enumerate}[label=(\alph*)]
	
		\item \label{item:ident:lra} $u$ is less risk-averse than $v$, i.e. for any alternative $x \in X$ and simple lottery $p \in \Delta^0(X)$, $u(x) \geq \mathrel{(>)} \int u \dd p$ implies $v(x) \geq \mathrel{(>)} \int v \dd p$.

		\item \label{item:ident:oo} There exist $\alpha>0$, $\beta \in \R$, and an extended CDF $F$ such that
		\begin{equation*}
			\alpha u(x) + \beta
			= \int \max\{v(x),k\} F(\dd k)
			\quad \text{for every $x \in X$}
		\end{equation*}
		and $F>0$ on $(\inftwo v(X),+\infty)$.
			
	\end{enumerate}
\end{theorem}

Another way of viewing \Cref{theorem:ident} is as a new characterisation of `less risk-averse than' for expected-utility preferences: $u$ is less risk-averse than $v$ if and only if choice according to $u$ can be rationalised as choice according to $v$ in the presence of an outside option. From this perspective, \Cref{theorem:ident} is a continuation of the enterprise of \hyperref[theorem:pratt]{Pratt's theorem} (\cpageref{theorem:pratt}).

A difference relative to \hyperref[theorem:pratt]{Pratt's theorem} is that \Cref{theorem:ident} restricts attention to risk attitudes $u,v : X \to \R$ that are bounded above and have $u$ Lipschitz with respect to $v$. We view both of these restrictions as mild.%
	\footnote{Boundedness above is standard, motivated by the St Petersburg paradox \parencite[see e.g.][lecture~9]{Rubinstein2012}. The Lipschitz requirement merely says that $u$ does not vary infinitely faster than $v$.}
We comment below on their role in the proof.

In property~\ref{item:ident:oo}, the requirement that $F>0$ on $(\inftwo v(X),+\infty)$ simply means that no alternative $x \in X$ is uniformly dominated by the outside option, provided we are in the typical case $\inftwo v(X) = \inf v(X)$.%
	\footnote{By `$x$ is uniformly dominated by the outside option', we mean that there exists an $\eps>0$ such that the outside option is a.s. worth strictly more than $v(x)+\eps$, i.e. $F(v(x)+\eps)=0$. In case $\inf v(X) < \inftwo v(X)$, `$F>0$ on $(\inftwo v(X),+\infty)$' means that no alternatives except the worst (those $x \in X$ with $v(x) = \inf v(X)$) are uniformly dominated by the outside option.}
Since $u$ is real-valued, property~\ref{item:ident:oo} implies that $\int_{(0,+\infty)} k F(\dd k) < +\infty$.

The fact that adding an outside option decreases risk-aversion (the implication from \ref{item:ident:oo} to \ref{item:ident:lra}) captures the natural intuition that gaining access to a fallback makes a decision-maker more willing to gamble. This is an old idea, at least in the special case of a deterministic outside option (see \cref{sec:intro:lit} above). For identification, this is a positive result: something that \emph{can} be learned about the `true' risk attitude $v$ from the effective risk attitude $u$ is that the former is more risk-averse than the latter.

The converse (the fact that \ref{item:ident:lra} implies \ref{item:ident:oo}) is less familiar: \emph{any} decrease of risk-aversion can be produced by making available an outside option (while leaving `true' risk attitude unchanged). For identification, this is a negative result: \emph{nothing more} can be learned about the `true' risk attitude $v$ from the effective risk attitude $u$ than that the former must be more risk-averse than the latter.

\begin{proof}[Sketch proof of \Cref{theorem:ident}]
	Suppose that \ref{item:ident:oo} holds. Then $u = \phi \circ v$ where
	\begin{equation*}
		\alpha \phi(\ell) + \beta
		= \int \max\{\ell,k\} F(\dd k)
		\quad \text{for every $\ell \in \co(v(X))$.}
	\end{equation*}
	Evidently $\phi$ is increasing and convex, since $\ell \mapsto \max\{\ell,k\}$ is for each $k \in [-\infty,+\infty)$. The fact that $F>0$ on $(\inftwo v(X),+\infty)$ ensures that $\phi$ is \emph{strictly} increasing on $\cotwo(v(X))$. Hence \ref{item:ident:lra} follows by \hyperref[theorem:pratt]{Pratt's theorem} (\cpageref{theorem:pratt}).

	For the converse, suppose that \ref{item:ident:lra} holds. Then by \hyperref[theorem:pratt]{Pratt's theorem}, there exists an increasing convex $\phi : \co(v(X)) \to \R$ that is strictly increasing on $\cotwo(v(X))$ and satisfies $u = \phi \circ v$. Write $\phi^+ : v(X) \setminus \{\sup v(X)\} \to \R$ for the right-hand derivative of $\phi$, and define $\lambda \coloneqq \lim_{k \nearrow \sup v(X)} \phi^+(k)$. Define $F : \R \to \R$ by $F \coloneqq \lim_{k \searrow \inf v(X)} \phi^+(k)$ on $(-\infty,\inf v(X)]$, $F \coloneqq \phi^+ / \lambda$ on $\co(v(X)) \setminus \{ \sup v(X) \}$, and $F \coloneqq 1$ on $[\sup v(X),+\infty)$. Then $F$ is an extended CDF; in particular, it is non-negative and increasing (since $\phi$ is increasing and convex) and satisfies $\lim_{k \nearrow +\infty} F(k) = 1$. By \Cref{proposition:fixedF}, there exist $\alpha>0$ and $\beta \in \R$ such that
	\begin{equation*}
		\alpha u(x) + \beta
		= \int \max\{v(x),k\} F(\dd k)
		\quad \text{for every $x \in X$.}
	\end{equation*}
	Finally, $F>0$ on $(\inftwo v(X),+\infty)$ since $\phi$ is strictly increasing on $\cotwo(v(X))$.
\end{proof}

The actual proof (\cref{app:ident_pf}) completes the above argument by filling in the gaps. For example, in the argument that \ref{item:ident:lra} implies \ref{item:ident:oo}, it must be shown that $\lambda < +\infty$ and (in order to invoke \Cref{proposition:fixedF}) that $\int_{(0,+\infty)} k F(\dd k) < +\infty$. The former follows immediately from the hypothesis that $u$ is Lipschitz with respect to $v$, while the latter follows from the fact that $v$ is bounded above, so $\int_{(0,+\infty)} k F(\dd k) \leq \sup v(X) < +\infty$. (This explains the role of these two hypotheses.)

\begin{remark}
	\label{remark:choquet}
	One could alternatively prove that \ref{item:ident:lra} implies \ref{item:ident:oo} via Choquet's theorem \parencite[see e.g.][p.~14]{Phelps2001}, as follows. By \hyperref[theorem:pratt]{Pratt's theorem}, \ref{item:ident:lra} implies that there exists an increasing convex $\phi : \co(v(X)) \to \R$ such that $u = \phi \circ v$. Since $u$ is Lipschitz with respect to $v$, $\phi$ is Lipschitz continuous with some constant $\alpha^{-1}>0$. A map $\psi : \co(v(X)) \to \R$ is called \emph{inflationary} if and only if $\psi(k) \geq k$ for every $k \in \co(v(X))$. Choose a $\beta \in \R$ large enough that $\alpha \phi + \beta$ is inflationary; this is possible since ($\phi$ is increasing and convex and) $v$ is bounded above. Then $\alpha \phi + \beta$ belongs to the set $\mathcal{F}$ of all increasing, convex, 1-Lipschitz and inflationary maps $\co(v(X)) \to \R$. Evidently $\mathcal{F}$ is convex, so by Choquet's theorem, there is a probability measure $\nu$ on $\mathcal{F}$ that is concentrated on the extreme points of $\mathcal{F}$ such that $\alpha \phi(\ell) + \beta = \int \psi(\ell) \nu(\dd \psi)$ for every $\ell \in \co(v(X))$. And the extreme points of $\mathcal{F}$ are exactly maps $\ell \mapsto \max\{\ell,k\}$ for $k \in [-\infty,+\infty)$, so by change of variable, $\alpha \phi(\ell) + \beta = \int \max\{\ell,k\} F(\dd k)$ for each $\ell \in \co(v(X))$, where $F$ is an extended CDF.%
		\footnote{Specifically, $F$ is given by $F(m) \coloneqq \nu(\{ \ell \mapsto \max\{\ell,k\} : k \in [-\infty,m] \})$ for each $m \in \R$.}
 	(This argument has various gaps that would need filling.)
\end{remark}

\begin{namedthm}[\Cref*{remark:oo_X} {\normalfont(continued from \cpageref{remark:oo_X})}.]
	\label{remark:oo_X_result}
	\Cref{theorem:ident} remains true if attention is restricted to extended CDFs $F$ arising from an $X$-valued outside option, modulo closure of $v(X)$: under the same hypotheses, \ref{item:ident:lra} is equivalent to

	\begin{enumerate}[label=(\alph*$'$)]

		\setcounter{enumi}{1}
	
		\item \label{item:ident:oo_supp} There exist $\alpha>0$, $\beta \in \R$, and an extended CDF $F$ such that
		\begin{equation*}
			\alpha u(x) + \beta
			= \int \max\{v(x),k\} F(\dd k)
			\quad \text{for every $x \in X$,}
		\end{equation*}
		$F>0$ on $(\inftwo v(X),+\infty)$, and $F$ is concentrated on $\cl(v(X)) \cup \{-\infty\}$.
			
	\end{enumerate}
	We prove this in \cref{app:remark_oo_pf}.
\end{namedthm}

\subsection{(Non-)identification of the outside-option distribution}
\label{sec:ident:F}

Whereas it is possible to learn something meaningful about the `true' risk attitude $v$ from the effective risk attitude $u$, essentially nothing can be learned about the outside-option distribution $F$:

\begin{proposition}
	\label{proposition:F_ident}
	For a non-empty set $X$, a function $u : X \to \R$ and an extended CDF $F$, the following are equivalent:

	\begin{enumerate}[label=(\Roman*)]
	
		\item \label{item:F_ident:finite} $\int_{(0,+\infty)} k F(\dd k) < +\infty$, and if $\int_{(-\infty,0]} k F(\dd k) > -\infty$ then $u$ is bounded below.

		\item \label{item:F_ident:oo} There exist $\alpha>0$, $\beta \in \R$, and a function $v : X \to \R$ such that
		\begin{equation*}
			\alpha u(x) + \beta
			= \int \max\{v(x),k\} F(\dd k)
			\quad \text{for every $x \in X$.}
		\end{equation*}

	\end{enumerate}
\end{proposition}

In other words, observation of the effective risk attitude $u$ tells us nothing about the outside-option distribution $F$ except that $\int_{(0,+\infty)} k F(\dd k) < +\infty$ and, if $u$ is unbounded below, that $\int_{(-\infty,0]} k F(\dd k) = -\infty$.

\section{Comparative statics}
\label{sec:mcs}

In this section, we consider comparative statics: how does the effective risk attitude $u$ vary with changes of the `true' risk attitude $v$ and of the outside-option distribution $F$? To see the forces at play, consider an example.

\begin{namedthm}[\Cref*{example:cara} {\normalfont(continued from \cpageref{example:cara,example:cara_ident}).}]
	\label{example:cara_mcs}
	Recall that $-\rho = -\sigma + \lambda$, where $-u''/u' = \rho \in \R$ and $-v''/v' = \sigma \in \R$ are the effective and `true' risk-aversion parameters and $G'/G = \lambda \in [0,+\infty)$ is the reverse hazard rate. Effective risk-aversion $\rho$ is increasing in the `true' risk-aversion parameter~$\sigma$. When the reverse hazard rate $\lambda$ increases, meaning that the outside-option distribution improves, effective risk-aversion $\rho$ decreases.
\end{namedthm}

In the example, the effective risk attitude $u$ becomes less risk-averse whenever the `true' risk attitude $v$ becomes less risk-averse, and whenever the outside-option distribution $F$ improves. Our comparative-statics theorem below shows that these findings are general (not specific to the example).

Beyond parametric examples, the right general notion of `improvement' of $F$ is the following: given $K \subseteq \R$ and extended CDFs $F$ and $\widehat{F}$, we say that $\widehat{F}$ is better than $F$ in the \emph{reverse hazard rate order on $K$} if and only if $F(\ell) \widehat{F}(k) \leq F(k) \widehat{F}(\ell)$ for all $k < \ell$ in $K$. If $K$ is an open interval and $F,\widehat{F}$ are continuously differentiable with $F > 0 < \widehat{F}$ on $(\inf K,+\infty)$, then this is equivalent to $\widehat{F}'/\widehat{F} \geq F'/F$ on $K$.

The reverse hazard rate order is a standard way of comparing distributions, used e.g. in the auctions literature.%
	\footnote{The idea of one bidder in an IPV auction being `stronger' than another is typically formalised as the former bidder's valuation distribution being better in the reverse hazard rate order (see \cite{Lebrun1998}, \cite{MaskinRiley2000}, and section~4.3.1 in \cite{Krishna2010}).}
Superiority of one distribution over another in the reverse hazard rate order implies first-order stochastic dominance, and is implied by superiority in the monotone likelihood ratio order. Reverse hazard rate ranking is equivalent to first-order stochastic dominance between, for each $k \in \R$, the conditional distributions obtained by conditioning on the ray $[-\infty,k]$ \parencite[see][p.~37]{ShakedShanthikumar2007}.

\begin{theorem}
	\label{theorem:mcs}
	Fix a non-empty set $X$, functions $u,v,\widehat{u},\widehat{v} : X \to \R$, an extended CDF $F$ with $F>0$ on $(\inftwo v(X),+\infty)$, and an extended CDF $\widehat{F}$ with $\widehat{F}>0$ on $\left(\inftwo \widehat{v}(X),+\infty\right)$. Assume that there exist $\alpha,\widehat{\alpha}>0$ and $\beta,\widehat{\beta} \in \R$ such that for every $x \in X$,
	\begin{align*}
		\alpha u(x) + \beta
		&= \int \max\{v(x),k\} F(\dd k)
		\quad \text{and}
		\\
		\widehat{\alpha} \widehat{u}(x) + \widehat{\beta}
		&= \int \max\left\{\widehat{v}(x),k\right\} \widehat{F}(\dd k) .
	\end{align*}
	Further suppose that $v(X)$ and $\widehat{v}(X)$ are Borel, that $F$ is concentrated on $v(X) \cup \{-\infty\}$, and that $\widehat{F}$ is concentrated on $\widehat{v}(X) \cup \{-\infty\}$.
	
	\begin{enumerate}[label=(\alph*)]

		\item \label{item:mcs:F} Suppose that $\widehat{v}=v$. Then $\widehat{u}$ is less risk-averse than $u$ if and only if $\widehat{F}$ is better than $F$ in the reverse hazard rate order on $v(X) \setminus \{\sup v(X)\}$.

		\item \label{item:mcs:v} Suppose that $F \circ v = \widehat{F} \circ \widehat{v}$. Then $\widehat{u}$ is less risk-averse than $u$ if and only if $\widehat{v}$ is less risk-averse than $v$.
	
	\end{enumerate}
\end{theorem}

\begin{namedthm}[\Cref*{remark:oo_X} {\normalfont(continued from \cpageref{remark:oo_X,remark:oo_X_result})}.]
	\label{remark:oo_X_mcs}
	The (Borel and) concentration assumptions mean that $F$ and $\widehat{F}$ both arise from $X$-valued outside options.%
		\footnote{The `if' half of part~\ref{item:mcs:v} remains true without these assumptions, with the same proof.}
	The supposition in part~\ref{item:mcs:v} means precisely that the shift from $F$ to $\widehat{F}$ holds fixed the distribution of this physical $X$-valued outside option, while changing the units of utility from $v$ to $\widehat{v}$. In case the alternatives are monetary ($X \subseteq \R$), what is held fixed is the extended CDF $G$ of the \emph{monetary} value of the outside option (which, recall, satisfies $F \circ v = G = \widehat{F} \circ \widehat{v}$). In general, the supposition in part~\ref{item:mcs:v} is equivalent to the existence of a \emph{single} probability measure $\mu$ on (some $\sigma$-algebra $\Sigma$ on) $X \cup \{\varnothing\}$ such that ($v$ and $\widehat{v}$ are $\Sigma$-measurable and) \emph{both} $(v,F)$ and $\bigl( \widehat{v}, \widehat{F} \bigr)$ satisfy \eqref{eq:F_from_mu} on~\cpageref{eq:F_from_mu}.
\end{namedthm}

In the special case of monetary alternatives and smooth primitives, \Cref{theorem:mcs} admits a simple and intuitive proof.

\begin{proof}[Proof of a special case of \Cref{theorem:mcs}]
	Assume that $X$ is an open convex subset of $\R$, that $v$ and $\widehat{v}$ are twice continuously differentiable with $v' > 0 < \widehat{v}'$, and that $F$ and $\widehat{F}$ are continuously differentiable. Let $G$ and $\widehat{G}$ be the (unique) extended CDFs concentrated on $X \union \{-\infty\}$ such that $G = F \circ v$ and $\widehat{G} = \widehat{F} \circ v$. Then by \Cref{proposition:fixedF} and the observation in \cref{sec:fixedF:monetary} (\cpageref{proposition:fixedF,sec:fixedF:monetary}), $u$ and $\widehat{u}$ are twice continuously differentiable with $u' > 0 < \widehat{u}'$,
	\begin{equation*}
		u''/u' = v''/v' + G'/G
		\quad \text{and} \quad
		\widehat{u}''/\widehat{u}' = \widehat{v}''/\widehat{v}' + \widehat{G}'/\widehat{G} .
	\end{equation*}
	Hence by \hyperref[theorem:pratt_diff]{Pratt's theorem (part~2}, \cpageref{theorem:pratt_diff}), 

	\begin{itemize}
	
		\item[\ref{item:mcs:F}] if $v=\widehat{v}$, then $\widehat{u}$ is less risk-averse than $u$ if and only if $\widehat{G}'/\widehat{G} \geq G'/G$, which is to say that $\widehat{G}$ is better than $G$ in the reverse hazard rate order on $X$, or equivalently (since $v$ and $\widehat{v}$ are strictly increasing) that $\widehat{F}$ is better than $F$ in the reverse hazard rate order on $v(X) \setminus \{\sup v(X)\}$, and

		\item[\ref{item:mcs:v}] if $G=\widehat{G}$, then $\widehat{u}$ is less risk-averse than $u$ if and only if $\widehat{v}''/\widehat{v}' \geq v''/v'$, i.e. if and only if $\widehat{v}$ is less risk-averse than $v$. \qedhere
	
	\end{itemize}
\end{proof}

This simple argument does not extend to the general case. The actual proof of \Cref{theorem:mcs} (\cref{app:mcs_pf}) is somewhat involved. It consists of four separate arguments (one each for the `if' and the `only if' parts of \ref{item:mcs:F} and \ref{item:mcs:v}), each of which makes intensive use of \hyperref[theorem:pratt]{Pratt's theorem}.

\section{What is special about outside options?}
\label{sec:special}

Adding an outside option to an expected-utility decision-maker's problem amounts to replacing each alternative $x$ with an ($x$-specific) lottery. In particular, in the outside-option model, each alternative $x$ is replaced by the lottery that returns the better of $x$ and a randomly-drawn outside option. Another model of this form is the \emph{background-risk model,} in which each alternative $x$ is replaced by the \emph{sum} of $x$ and randomly-drawn `wealth'.

By \Cref{theorem:ident}, adding an outside option reduces risk-aversion. In this section, we show that among all transformations that replace each alternative with a lottery, those that amount to adding an outside option are the \emph{only} ones that reduce risk-aversion. Adding background risk does \emph{not} generally decrease (or increase) risk-aversion, for example.%
	\footnote{This is known: see the references in \cref{footnote:special_background} (\cpageref{footnote:special_background}).}

To formalise these ideas, we consider a non-empty set $X$ of alternatives and a collection $(\mu_x)_{x \in X}$ of probability measures on (some $\sigma$-algebra on) $X$. When the decision-maker draws $x \in X$ from a lottery, the alternative which she actually takes home is a random draw from $\mu_x$. Thus if her `true' risk attitude is $v : X \to \R$, then her effective risk attitude is $x \mapsto \int v \dd \mu_x$. Given a class $\mathcal{V}$ of possible `true' risk attitudes $v$, our question is for which collections $(\mu_x)_{x \in X}$ it holds that for every $v \in \mathcal{V}$, $x \mapsto \int v \dd \mu_x$ is less risk-averse than $v$.

Two maps $v,\widehat{v} : X \to \R$ are called \emph{ordinally equivalent} if and only if for all $x,y \in X$, $v(x) \geq \mathrel{(>)} v(y)$ implies $\widehat{v}(x) \geq \mathrel{(>)} \widehat{v}(y)$. Since one map $X \to \R$ can be less risk-averse than another only if the two are ordinally equivalent,%
	\footnote{This is the `\ref{item:pratt:lra} implies \ref{item:pratt:curv}\ref{item:pratt:curv:ordequiv}' part of \hyperref[theorem:pratt]{Pratt's theorem} (\cpageref{theorem:pratt}).}
we restrict attention to classes $\mathcal{V}$ whose members are all mutually ordinally equivalent. We impose no further restrictions; thus the classes $\mathcal{V}$ that we consider are exactly the ordinal equivalence classes of maps $X \to \R$.

When considering any given such ordinal equivalence class $\mathcal{V}$, it is convenient to fix an arbitrary $v_0 \in \mathcal{V}$ and embed the alternatives $X$ in $\R$ via $x \mapsto v_0(x)$. Then $\mathcal{V}$ is the set of all strictly increasing functions $X \to \R$.

With these conventions, our question becomes: given a non-empty Borel set $X \subseteq \R$, for which collections $(G_x)_{x \in X}$ of CDFs concentrated on $X$ does it hold that for every strictly increasing $v : X \to \R$, the map $x \mapsto \int v \dd G_x$ is less risk-averse than $v$? The following theorem answers this question.

\begin{theorem}
	\label{theorem:special}
	For a non-empty Borel set $X \subseteq \R$ with $\sup X \notin X$ and a collection $(G_x)_{x \in X}$ of CDFs concentrated on $X$, the following are equivalent:

	\begin{enumerate}[label=(\alph*)]
		
		\item \label{item:special:lra} For any bounded and strictly increasing function $v : X \to \R$, the map $x \mapsto \int v \dd G_x$ is less risk-averse than $v$.
		
		\item \label{item:special:oo} 
		There exists an extended CDF $G$ concentrated on $X \cup \{-\infty\}$ with $G>0$ on $(\inf X,+\infty)$ such that for every bounded and strictly increasing function $v : X \to \R$, there exist $\alpha>0$ and $\beta \in \R$ such that
		\begin{equation*}
			\alpha \int v \dd G_x + \beta
			= \int \max\{v(x),v(y)\} G(\dd y)
			\quad \text{for every $x \in X$,}
		\end{equation*}
		where by convention $v(-\infty) \coloneqq -\infty$.

	\end{enumerate}
\end{theorem}

Thus the only transformations replacing alternatives with lotteries which decrease risk-aversion are those that amount to adding an outside option.

The restriction in \Cref{theorem:special} to bounded $v : X \to \R$ is mild; its role is to ensure that all integrals are finite. The assumption that $\sup X \notin X$ is for simplicity; without it, \ref{item:special:oo} must be replaced with a slightly more complicated property, as we show in \cref{app:special_greatest}.

The fact that \ref{item:special:oo} implies \ref{item:special:lra} follows immediately from \Cref{theorem:ident} (and a change of variable). The converse does not: if \ref{item:special:lra} holds, then what \Cref{theorem:ident} (and a change of variable) delivers is, separately for each bounded and strictly increasing $v : X \to \R$, an extended CDF $G^v$ concentrated on $X \cup \{-\infty\}$ with $G^v>0$ on $(\inf X,+\infty)$ for which there are $\alpha>0$ and $\beta \in \R$ such that
\begin{equation*}
	\alpha \int v \dd G_x + \beta
	= \int \max\{v(x),v(y)\} G^v(\dd y)
	\quad \text{for every $x \in X$.}
\end{equation*}
The substance of \Cref{theorem:special} is the assertion that if \ref{item:special:lra} holds, then these $G^v$s can be chosen in a $v$-independent way. We prove this in \cref{app:special_pf}.

\begin{appendices}

\crefalias{section}{appsec}
\crefalias{subsection}{appsec}
\crefalias{subsubsection}{appsec}

\section{More on \texorpdfstring{\hyperref[theorem:pratt]{Pratt's theorem (part~1,} \cpageref{theorem:pratt})}{Pratt's theorem (part 1, p. \pageref{theorem:pratt})}}
\label{app:pratt_phis}

To see why property~\ref{item:pratt:trans} in \hyperref[theorem:pratt]{Pratt's theorem (part~1)} does not assert that $\phi$ can be chosen to be continuous, or to be strictly increasing, consider the following two examples. (Recall from \cpageref{definition:cotwo} the definition of `$\cotwo$', and that $A \subseteq \cotwo(A) \subseteq \co(A)$ for any set $A \subseteq \R$.)

\begin{example}
	\label{example:pratt_disc}
	Consider $X \coloneqq [0,1]$ and $u,v : X \to \R$, where $v$ is the identity, $u=v$ on $[0,1)$, and $u(1)=2$. Then $u$ is less risk-averse than $v$, but the only $\phi : \co(v(X)) \to \R$ that satisfies $u = \phi \circ v$ is discontinuous at $\max v(X) = 1$.
\end{example}

\begin{example}
	\label{example:pratt_str}
	Consider $X \coloneqq [1,2]$ and $u,v : X \to \R$, where $u$ is the identity, $v(1)=0$, and $v=u$ on $(1,2]$. Then $u$ is less risk-averse than $v$, but the only increasing $\phi : \co(v(X)) \to \R$ that satisfies $u = \phi \circ v$ is constant on $[\inf v(X), \inftwo v(X)] = [0,1]$.
\end{example}

\Cref{example:pratt_str} shows that the strict monotonicity of $\phi$ on $v(X)$ in property~\ref{item:pratt:trans} in \hyperref[theorem:pratt]{Pratt's theorem} cannot be strengthened to strict monotonicity on $\co(v(X))$. It can, however, be strengthened to strict monotonicity on $\cotwo(v(X))$:

\begin{lemma}
	\label{lemma:greatest_phi}
	Fix a non-empty set $X$ and functions $u,v : X \to \R$, and let $\Phi$ be the set of all increasing convex functions $\phi : \co(v(X)) \to \R$ that are strictly increasing on $v(X)$ and satisfy $u = \phi \circ v$. If $\Phi$ is not empty, then it has a pointwise greatest element, which is strictly increasing on $\cotwo(v(X))$ and affine on each maximal interval of $\co(v(X)) \setminus v(X)$.
\end{lemma}

We use \Cref{lemma:greatest_phi} in several proofs: the strict monotonicity conclusion is used in the proofs of \hyperref[theorem:pratt]{Pratt's theorem} (\cref{app:pratt_pf}), \Cref{theorem:ident} (\cref{app:ident_pf}) and \Cref{theorem:mcs} (\cref{app:mcs_pf}), while the affineness conclusion is used in \cref{app:remark_oo_pf} to prove \Cref{remark:oo_X} on \cpageref{remark:oo_X_result}.

\begin{proof}
	Assume that $\Phi$ is non-empty. Note that there is exactly one function $\phi_0 : v(X) \to \R$ such that $u = \phi_0 \circ v$, and that this $\phi_0$ is strictly increasing. Define $\phi : \co(v(X)) \to \R$ by $\phi(k) \coloneqq \sup_{\psi \in \Phi} \psi(k)$ for each $k \in \co(v(X))$. By inspection, $\phi$ is increasing and convex, and is strictly increasing on $v(X)$ since $\phi=\phi_0$ on $v(X)$; thus $\phi \in \Phi$. Obviously $\phi \geq \psi$ for every $\psi \in \Phi$. Since each $\psi \in \Phi$ is convex and satisfies $\psi = \phi_0$ on $v(X)$, $\phi$ is affine on each maximal interval of $\co(v(X)) \setminus v(X)$.

	To show that $\phi$ is strictly increasing on $\cotwo(v(X))$, fix any $k'<\ell'$ in $\cotwo(v(X))$; we must show that $\phi(k') < \phi(\ell')$. It suffices to find $k<\ell$ in $\co(v(X))$ such that $k \leq k'$, $\ell \leq \ell'$, and $\phi(k)<\phi(\ell)$, since then
	\begin{equation*}
		\phi(\ell') - \phi(k')
		\geq (\ell'-k') \frac{\phi(\ell) - \phi(k)}{\ell-k}
		> 0 ,
	\end{equation*}
	where the weak inequality holds since $\phi$ is convex. Write $m_1 \coloneqq \inf v(X)$ and $m_2 \coloneqq \inf( v(X) \setminus \{m_1\} )$, and note that $k' \geq m_1 < \ell' \geq m_2$.

	We consider four cases. In the first three, we find $k<\ell$ in $v(X)$ such that $k \leq k'$ and $\ell \leq \ell'$; then $\phi(k)<\phi(\ell)$ since $\phi$ is strictly increasing on $v(X)$. In the final case, we directly choose $k<\ell$ in $\co(v(X))$ to satisfy $\phi(k)<\phi(\ell)$.

	\smallskip

	\noindent
	\emph{Case~1: $m_1 \notin v(X)$.} Here $k' > m_1$ and $\ell' > m_1 = m_2$, so choosing $k<\ell$ in $v(X)$ sufficiently close to $m_1$ ensures that $k \leq k'$ and $\ell \leq \ell'$.

	\smallskip

	\noindent
	\emph{Case~2: $v(X) \ni m_1 = m_2$.} 
	Here $m_1 \in \cl( v(X) \setminus \{m_1\} )$, so choosing $k \coloneqq m_1$ and $\ell \in v(X) \setminus \{m_1\}$ sufficiently close to $m_1$ ensures that $k \leq k'$ and $\ell \leq \ell'$.

	\smallskip

	\noindent
	\emph{Case~3: $m_1 < m_2 \notin v(X)$.}
	Here $\varnothing \neq \co(v(X)) \setminus \cotwo(v(X)) = (m_1,m_2]$, whence $m_1 \in v(X)$ and $\ell'>m_2$, so choosing $k \coloneqq m_1$ and $\ell \in v(X) \setminus [m_1,m_2]$ sufficiently close to $m_2$ ensures that $k \leq k'$ and $\ell \leq \ell'$.

	\smallskip

	\noindent
	\emph{Case~4: $m_1 < m_2 \in v(X)$.}
	Here $m_1 \in v(X)$, so $\phi(m_1)<\phi(m_2)$, which since $\phi$ is affine on $[m_1,m_2]$ implies that $\phi$ is strictly increasing on $[m_1,m_2]$, so that $k \coloneqq m_1$ and $\ell \coloneqq \min\{\ell',m_2\}$ satisfy $\phi(k)<\phi(\ell)$.
\end{proof}

\section{Proof of \texorpdfstring{\hyperref[theorem:pratt]{Pratt's theorem (part~1}, \cpageref{theorem:pratt})}{Pratt's theorem (part 1, p. \pageref{theorem:pratt})}}
\label{app:pratt_pf}

We shall prove that \ref{item:pratt:trans} implies \ref{item:pratt:lra} implies \ref{item:pratt:curv} implies \ref{item:pratt:trans}.

To prove that \ref{item:pratt:trans} implies \ref{item:pratt:lra}, suppose there exists an increasing convex function $\psi : \co(v(X)) \to \R$ that is strictly increasing on $v(X)$ and satisfies $u = \psi \circ v$. Then by \Cref{lemma:greatest_phi} in \cref{app:pratt_phis}, there exists an increasing convex function $\phi : \co(v(X)) \to \R$ that is strictly increasing on $\cotwo(v(X))$ and satisfies $u = \phi \circ v$. Fix an alternative $x \in X$ and a simple lottery $p \in \Delta^0(X)$, and suppose that $\int v \dd p \geq \mathrel{(>)} v(x)$; we must show that $\int u \dd p \geq \mathrel{(>)} u(x)$. If $\int v \dd p \in \cotwo(v(X))$, then
\begin{equation*}
	\int u \dd p
	= \int (\phi \circ v) \dd p
	\geq \phi\left( \int v \dd p \right)
	\geq \mathrel{(>)} \phi(v(x))
	= u(x) ,
\end{equation*}
where the first inequality holds (by Jensen's inequality) since $\phi$ is convex, and the second holds since $\phi$ is strictly increasing on $\cotwo(v(X)) \supseteq v(X) \ni v(x)$. If instead $\int v \dd p \notin \cotwo(v(X))$, then
\begin{equation*}
	\int v \dd p
	\in \co(v(X)) \setminus \cotwo(v(X))
	= \bigl( \inf v(X), \inf\bigl( v(X) \setminus \{\inf v(X)\} \bigr) \bigr] ,
\end{equation*}
so writing $Y \coloneqq \{ y \in X : v(y) = \inf v(X) \}$, we see that $\int v \dd p \geq \mathrel{(>)} v(x)$ implies $x \in Y$ (and $p(X \setminus Y)>0$), whence
\begin{align*}
	\int u \dd p
	={}& p(Y) \phi(v(x)) 
	+ \int_{X \setminus Y} (\phi \circ v) \dd p
	\\
	\geq \mathrel{(>)}{}& p(Y) \phi(v(x))
	+ p(X \setminus Y) \phi(v(x))
	= u(x) 
\end{align*}
since $\phi$ is strictly increasing on $v(X)$.

To prove that \ref{item:pratt:lra} implies \ref{item:pratt:curv}, suppose that $u$ is less risk-averse than $v$. It follows immediately (by considering degenerate lotteries $p \in \Delta^0(X)$) that property~\ref{item:pratt:curv}\ref{item:pratt:curv:ordequiv} holds.
To show that property~\ref{item:pratt:curv}\ref{item:pratt:curv:compress} holds, suppose toward a contradiction that it does not: there are $x,y,z \in X$ such that $u(x) < u(y) < u(z)$ and
\begin{equation*}
	\frac{u(z)-u(y)}{u(y)-u(x)}
	< \frac{v(z)-v(y)}{v(y)-v(x)} .
\end{equation*}
By replacing $u$ with $\alpha u + \beta$ for some $\alpha > 0$ and $\beta \in \R$ if necessary, we may assume without loss of generality that $u(x) = v(x)$ and $u(y) = v(y)$, so that $u(z) < v(z)$.
Define a simple lottery $p \in \Delta^0(X)$ by $p(x) \coloneqq [ u(z) - u(y) ] / [ u(z) - u(x) ]$, $p(z) \coloneqq 1-p(x)$, and $p(w) \coloneqq 0$ for every $w \in X \setminus \{x,z\}$.
Then $u(y) = \int u \dd p$ and
\begin{equation*}
	v(y)
	= u(y)
	= p(x) u(x) + p(z) u(z)
	< p(x) v(x) + p(z) v(z)
	= \int v \dd p ,
\end{equation*}
a contradiction with the fact that $u$ is less risk-averse than $v$.

To prove that \ref{item:pratt:curv} implies \ref{item:pratt:trans}, suppose that $u$ satisfies properties~\ref{item:pratt:curv}\ref{item:pratt:curv:ordequiv} and \ref{item:pratt:curv}\ref{item:pratt:curv:compress}; we must identify an increasing convex function $\phi : \co(v(X)) \to \R$ that is strictly increasing on $v(X)$ and satisfies $u = \phi \circ v$. By property~\ref{item:pratt:curv}\ref{item:pratt:curv:ordequiv}, there exists a strictly increasing $\psi : v(X) \to \R$ such that $u = \psi \circ v$.
Define $\widebar{\psi} : \cl(v(X)) \cap \co(v(X)) \to \R$ by
\begin{equation*}
	\widebar{\psi}(k)
	\coloneqq
	\begin{cases}
		\psi(k) & \text{if $k \in v(X)$} \\
		\lim_{\ell \to k} \psi(\ell) & \text{if $k \in [ \cl(v(X)) \cap \co(v(X)) ] \setminus v(X)$,} 
	\end{cases}
\end{equation*}
where the limit exists (in $\R$) by the monotonicity of $\psi$ and property~\ref{item:pratt:curv}\ref{item:pratt:curv:compress}.%
	\footnote{Fix any $k \in [ \cl(v(X)) \cap \co(v(X)) ] \setminus v(X)$. Since $k \in \cl(v(X))$, there is a monotone sequence in $v(X)$ that converges to $k$, which since $\psi$ is increasing implies that either the left-hand limit $\psi(k-)$ or the right-hand limit $\psi(k+)$ must exist in $\R \cup \{-\infty,+\infty\}$. Since $k \in \co(v(X))$, there are $m_0,m_1 \in v(X)$ such that $m_0 \leq k \leq m_1$. Since $\psi$ is increasing, it is bounded on $[m_0,m_1] \intersect v(X)$. Hence $\psi(k-)$ is finite if it exists, and likewise for $\psi(k+)$.

	It remains only to show that if $\psi(k-)$ and $\psi(k+)$ both exist, then they are equal. We have $\psi(k-) \leq \psi(k+)$ since $\psi$ is increasing. To show that $\psi(k-) \geq \psi(k+)$, suppose toward a contradiction that $\psi(k-) < \psi(k+)$, and fix an $m \in v(X)$ such that $k<m$. Then we can choose $\ell,k' \in v(X)$ arbitrarily close to $k$ and satisfying $\ell<k<k'<m$, and by doing so we may make $[ \psi(k') - \psi(\ell) ] / [k'-\ell]$ arbitrarily large. Since $\psi$ is increasing, it is bounded on a neighbourhood of $k$, so $\ell,k' \in v(X)$ can be chosen so that $[ \psi(m) - \psi(k') ] / [m-k']$ is bounded. Hence $\ell,k' \in v(X)$ can be chosen so that $[ \psi(k') - \psi(\ell) ] / [k'-\ell] > [ \psi(m) - \psi(k') ] / [m-k']$, a contradiction with property~\ref{item:pratt:curv}\ref{item:pratt:curv:compress}.}
Let $\phi$ be the (unique) function $\co(v(X)) \to \R$ that matches $\widebar{\psi}$ on $\cl(v(X)) \cap \co(v(X))$ and is affine on the closure of each maximal interval of $\co(v(X)) \setminus v(X)$.
Evidently $\phi$ is increasing, and $\phi$ is convex since by property~\ref{item:pratt:curv}\ref{item:pratt:curv:compress},
\begin{equation*}
	\frac{\phi(\ell)-\phi(k)}{\ell-k}
	\leq \frac{\phi(m)-\phi(\ell)}{m-\ell}
	\quad \text{for all $k < \ell < m$ in $\co(v(X))$.}
\end{equation*}
Since $\phi=\psi$ on $v(X)$, $\phi$ is strictly increasing on $v(X)$, and $u = \phi \circ v$.
\qed

\section{Proof of \texorpdfstring{\hyperref[theorem:pratt_diff]{Pratt's theorem (part~2}, \cpageref{theorem:pratt_diff})}{Pratt's theorem (part 2, p. \pageref{theorem:pratt_diff})}}
\label{app:pratt_diff_pf}

Note that property~\ref{item:pratt:curv}\ref{item:pratt:curv:ordequiv} holds since $u$ and $v$ are strictly increasing (as $u' > 0 < v'$). Hence by \hyperref[theorem:pratt]{Pratt's theorem (part~1}, \cpageref{theorem:pratt}), what must be shown is that property~\ref{item:pratt:curv}\ref{item:pratt:curv:compress} holds if and only if $u''/u' \geq v''/v'$.

Suppose that property~\ref{item:pratt:curv}\ref{item:pratt:curv:compress} holds. Then for any $x<y<z<w$ in $X$,
\begin{align*}
	\frac{u(w)-u(z)}{u(y)-u(x)}
	&= \frac{u(w)-u(z)}{u(z)-u(y)}
	\times \frac{u(z)-u(y)}{u(y)-u(x)}
	\\
	&\geq \frac{v(w)-v(z)}{v(z)-v(y)}
	\times \frac{v(z)-v(y)}{v(y)-v(x)}
	= \frac{v(w)-v(z)}{v(y)-v(x)} .
\end{align*}
Hence for each $x \in X$,
\begin{align*}
	\frac{u''(x)}{u'(x)}
	&= \left. \frac{\dd}{\dd y} \ln( u'(y) ) \right|_{y=x}
	= \lim_{\eps \searrow 0} \frac{1}{\eps}
	\ln\left( \frac{u'(x+\eps)}{u'(x)} \right)
	\\
	&= \lim_{\eps \searrow 0}
	\frac{1}{\eps} \ln\left( \frac
	{ \lim_{\delta \searrow 0} \frac{1}{\delta}
	\left[ u(x+\eps+\delta) - u(x+\eps) \right] }
	{ \lim_{\delta \searrow 0} \frac{1}{\delta}
	\left[ u(x+\delta) - u(x) \right] }
	\right)
	\\
	&= \lim_{\eps \searrow 0}
	\lim_{\delta \searrow 0} 
	\frac{1}{\eps} \ln\left( \frac
	{ u(x+\eps+\delta) - u(x+\eps) }
	{ u(x+\delta) - u(x) }
	\right) 
	\\
	&\geq \lim_{\eps \searrow 0}
	\lim_{\delta \searrow 0} 
	\frac{1}{\eps} \ln\left( \frac
	{ v(x+\eps+\delta) - v(x+\eps) }
	{ v(x+\delta) - v(x) }
	\right)
	= \frac{v''(x)}{v'(x)} .
\end{align*}

Conversely, suppose that $u''/u' \geq v''/v'$. Since $u' > 0 < v'$, we have
\begin{equation*}
	u'(w)
	= u'\left( y \right)
	\exp\left( \int_y^w \frac{u''}{u'} \right) 
	\quad \text{and} \quad
	v'(w)
	= v'\left( y \right)
	\exp\left( \int_y^w \frac{v''}{v'} \right) 
\end{equation*}
for any $y,w \in X$. Hence by the fundamental theorem of calculus, it holds for any $x<y<z$ in $X$ that
\begin{equation*}
	\frac{ u(z) - u(y) }{ u(y) - u(x) }
	= \frac
	{ \int_y^z
	\exp\left( \int_y^w \frac{u''}{u'} \right)
	\dd w }
	{ \int_x^y
	\exp\left( - \int_w^y \frac{u''}{u'} \right)
	\dd w } 
	\geq \frac
	{ \int_y^z
	\exp\left( \int_y^w \frac{v''}{v'} \right)
	\dd w }
	{ \int_x^y
	\exp\left( - \int_w^y \frac{v''}{v'} \right)
	\dd w } 
	= \frac{ v(z) - v(y) }{ v(y) - v(x) } .
\end{equation*}
\qed

\section{A differentiation lemma}
\label{app:lemma_chi}

The following lemma is used in the proofs below of \Cref{theorem:ident,theorem:mcs} and \Cref{proposition:fixedF,proposition:F_ident} (\cref{app:fixedF_pf,app:ident_pf,app:F_ident_pf,app:mcs_pf}).

\begin{lemma}
	\label{lemma:chi}
	Let $F$ be an extended CDF, and define $\chi : \R \to (-\infty,+\infty]$ by $\chi(\ell) \coloneqq \int \max\{\ell,k\} F(\dd k)$ for every $\ell \in \R$.

	\begin{itemize}
	
		\item If $\chi<+\infty$, then $\chi$ is 1-Lipschitz, and differentiable at each $\ell \in \R$ at which $F$ has no atom, with derivative $\chi'(\ell) = F(\ell)$.

		\item $\chi<+\infty$ holds if and only if there exists an $\ell \in \R$ such that $\chi(\ell)<+\infty$, which holds if and only if $\int_{(0,+\infty)} k F(\dd k) < +\infty$.
	
	\end{itemize}
\end{lemma}

\begin{proof}
	The second claim (characterising $\chi<+\infty$) holds by inspection. For the first claim, note that for each $k \in [-\infty,+\infty)$, the map $\R \to \R$ given by $\ell \mapsto \max\{\ell,k\}$ is 1-Lipschitz, with derivative equal to $0$ on $(-\infty,k)$ and to $1$ on $(k,+\infty)$.
	Hence if $\chi<+\infty$, then $\chi$ is 1-Lipschitz, and differentiable at each $\ell \in \R$ at which $F$ has no atom, with the derivative being
	\begin{equation*}
		\chi'(\ell)
		= \int \left. \frac{\dd}{\dd m} \max\{m,k\} \right|_{m=\ell} F(\dd k)
		\\
		= \int \1_{[-\infty,\ell]} \dd F
		= F(\ell)  ,
	\end{equation*}
	where the first equality holds by the bounded convergence theorem (applicable since for every $k \in [-\infty,+\infty)$, $m \mapsto \max\{m,k\}$ is 1-Lipschitz).
\end{proof}

\section{Proof of \texorpdfstring{\Cref{proposition:fixedF} (\cpageref{proposition:fixedF})}{Proposition~\ref{proposition:fixedF} (p.~\pageref{proposition:fixedF})}}
\label{app:fixedF_pf}

The result is trivial if $v$ is constant. Assume for the remainder that $v$ is non-constant. (This implies that $\abs*{X} \geq 2$, of course.) Define $\chi : \co(v(X)) \to (-\infty,+\infty]$ by $\chi(\ell) \coloneqq \int \max\{\ell,k\} F(\dd k)$ for every $\ell \in \co(v(X))$. We shall make use of \Cref{lemma:chi} from \cref{app:lemma_chi}.

Evidently \ref{item:fixedF:rep} holds if and only if there exist $\alpha>0$ and $\beta \in \R$ such that $\psi \coloneqq (\chi-\beta)/\alpha$ satisfies $u = \psi \circ v$.
Thus if \ref{item:fixedF:rep} holds, then $\psi(v(x)) = u(x) < +\infty$ for some (indeed, any) $x \in X$, so $\chi<+\infty$ by \Cref{lemma:chi}, which by \Cref{lemma:chi} implies \ref{item:fixedF:trans} with $\lambda \coloneqq 1/\alpha$ and $\phi \coloneqq \psi$. Conversely, suppose that \ref{item:fixedF:trans} holds, and let $\alpha \coloneqq 1/\lambda$. Then by \Cref{lemma:chi}, $\chi<+\infty$, and the (real-valued) functions $\alpha \phi$ and $\chi$ are absolutely continuous (the former by hypothesis, the latter since it is Lipschitz) with a.e. equal derivatives. Hence by the (Lebesgue) fundamental theorem of calculus, there exists a $\beta \in \R$ such that $\alpha \phi + \beta = \chi$, which is to say that \ref{item:fixedF:rep} holds (with $\psi = \phi$).
\qed

\section{Proof of \texorpdfstring{\Cref{theorem:ident} (\cpageref{theorem:ident})}{Theorem~\ref{theorem:ident} (p.~\pageref{theorem:ident})}}
\label{app:ident_pf}

The result is trivial if $v$ is constant (note that in this case, \ref{item:ident:lra} implies that $u$ is constant). Assume for the remainder that $v$ is non-constant. (This implies that $\abs*{X} \geq 2$, of course.) We shall make use of \Cref{lemma:greatest_phi,lemma:chi} from \cref{app:pratt_phis,app:lemma_chi}.

Suppose that \ref{item:ident:oo} holds. Then $u = \phi \circ v$ where $\phi \coloneqq (\chi-\beta)/\alpha$ and $\chi(\ell) \coloneqq \int \max\{\ell,k\} F(\dd k)$ for every $\ell \in \co(v(X))$. By \hyperref[theorem:pratt]{Pratt's theorem} (\cpageref{theorem:pratt}), it suffices to show that $\phi$ is increasing and convex, and strictly increasing on $v(X)$. That $\phi$ is increasing and convex follows from the fact that $\chi$ is, since for each $k \in [-\infty,+\infty)$, $\ell \mapsto \max\{\ell,k\}$ is increasing and convex.

To show that $\phi$ is strictly increasing on $v(X)$, note that since $u<+\infty$, we have $\phi<+\infty$, and thus $\chi<+\infty$ by \Cref{lemma:chi}. Hence by \Cref{lemma:chi}, $\phi$ is Lipschitz continuous with $\phi' = \alpha^{-1} F$ a.e., so by the (Lebesgue) fundamental theorem of calculus, for any $k<\ell$ in $\cotwo(v(X)) \supseteq v(X)$, $\phi(\ell)-\phi(k) = \alpha^{-1} \int_k^\ell F > 0$ since $F>0$ on $(\inftwo v(X), +\infty )$.

For the converse, suppose that \ref{item:ident:lra} holds. Then by \hyperref[theorem:pratt]{Pratt's theorem} and \Cref{lemma:greatest_phi}, there exists an increasing convex $\phi : \co(v(X)) \to \R$ that is strictly increasing on $\cotwo(v(X))$ and satisfies $u = \phi \circ v$. Since $u$ is Lipschitz with respect to $v$, $\phi$ is Lipschitz continuous. Let $\phi^+ : \co(v(X)) \setminus \{\sup v(X)\} \to \R$ denote the right-hand derivative of $\phi$; it is well-defined since $\phi$ is convex, and finite (indeed, bounded) since $\phi$ is Lipschitz continuous. Let $\lambda \coloneqq \lim_{k \nearrow \sup v(X)} \phi^+(k)$. Note that $\lambda>0$,%
	\footnote{\label{footnote:oo_pf_0}Since $\phi$ is increasing, $\lambda<0$ is impossible, and if $\lambda=0$ then $\phi^+=0$ on $\co(v(X)) \setminus \{\sup v(X)\}$, in which case ($\phi$ being Lipschitz continuous) we would have for any $k<\ell$ in $\co(v(X))$ that $\phi(\ell)-\phi(k) = \int_k^\ell \phi^+ = 0$ by the (Lebesgue) fundamental theorem of calculus, a contradiction with the fact that $\phi$ is strictly increasing on $\cotwo(v(X))$.}
and that $\lambda<+\infty$ since $\phi$ is Lipschitz continuous.

Define $F : \R \to \R$ by $F \coloneqq \lim_{k \searrow \inf v(X)} \phi^+(k)$ on $(-\infty,\inf v(X)]$, $F \coloneqq \phi^+ / \lambda$ on $\co(v(X)) \setminus \{ \sup v(X) \}$, and $F \coloneqq 1$ on $[\sup v(X),+\infty)$. Since $\phi$ is increasing and convex, $F$ is an extended CDF: it is non-negative, increasing, right-continuous \parencite[by Theorem~24.1 in][]{Rockafellar1970}, and has $\lim_{k \nearrow +\infty} F(k) = 1$. Since $\sup v(X)<+\infty$ (as $v$ is bounded above) and $F = 1$ on $[\sup v(X),+\infty)$, $\int_{(0,+\infty)} k F(\dd k) \leq \sup v(X) < +\infty$. Hence by \Cref{proposition:fixedF}, there exist $\alpha>0$ and $\beta \in \R$ such that
\begin{equation*}
	\alpha u(x) + \beta
	= \int \max\{v(x),k\} F(\dd k)
	\quad \text{for every $x \in X$.}
\end{equation*}
Finally, since $\phi$ is Lipschitz continuous and is strictly increasing on $\cotwo(v(X))$, we have $F>0$ on $(\inftwo v(X),+\infty)$.%
	\footnote{The argument is the same as that in \cref{footnote:oo_pf_0}.}
This establishes \ref{item:ident:oo}.
\qed

\section{Proof of \texorpdfstring{{\hyperref[remark:oo_X_result]{\Cref*{remark:oo_X}} on \cpageref{remark:oo_X_result}}}{Remark~\ref{remark:oo_X} on p. \pageref{remark:oo_X_result}}}
\label{app:remark_oo_pf}

Suppose that property~\ref{item:ident:lra} in \Cref{theorem:ident} (\cpageref{item:ident:lra}) holds; we shall deduce property~\ref{item:ident:oo_supp}. By \hyperref[theorem:pratt]{Pratt's theorem}, the set $\Phi$ defined in \Cref{lemma:greatest_phi} (\cref{app:pratt_phis}) is non-empty, so by \Cref{lemma:greatest_phi}, $\Phi$ has a pointwise greatest element $\phi$. Now follow the proof in \cref{app:ident_pf} of the `\ref{item:ident:lra} implies \ref{item:ident:oo}' part of \Cref{theorem:ident}: by \Cref{proposition:fixedF} (\cpageref{proposition:fixedF}), property~\ref{item:ident:oo} in \Cref{theorem:ident} (\cpageref{item:ident:oo}) holds for the extended CDF $F$ that is constant on $(-\infty,\inf v(X)]$, proportional to the right-hand derivative $\phi^+$ on $\co(v(X)) \setminus \{ \sup v(X) \}$, and equal to $1$ on $[\sup v(X),+\infty)$. For each maximal interval $I$ of $\co(v(X)) \setminus v(X)$, $\phi^+=0$ on $\interior(I)$ by \Cref{lemma:greatest_phi}, so $\int_{\co(v(X)) \setminus \cl(v(X))} \dd F = 0$. Hence
\begin{multline*}
	\int_{\cl(v(X)) \cup \{-\infty\}} \dd F
	\\
	= 1 - \int_{\co(v(X)) \setminus \cl(v(X))} \dd F
	- \int_{(-\infty,\inf v(X)]} \dd F
	- \int_{[\sup v(X),+\infty)} \dd F
	= 1 . 
\end{multline*}
\qed

\section{Proof of \texorpdfstring{\Cref{proposition:F_ident} (\cpageref{proposition:F_ident})}{Theorem~\ref{proposition:F_ident} (p. \pageref{proposition:F_ident})}}
\label{app:F_ident_pf}

We will use the following simple result.

\begin{observation}
	\label{observation:F_integrals}
	$\int_{(-\infty,0]} k F(\dd k) > -\infty$ if and only if $\int_{-\infty}^0 F < +\infty$.
\end{observation}

\begin{proof}[Proof of \Cref{observation:F_integrals}]
	\renewcommand{\qedsymbol}{$\square$}
	Write $I \coloneqq \int_{(-\infty,0]} k F(\dd k)$. We have
	\begin{equation*}
		\int_{(-\ell,0]} k F(\dd k)
		= \ell F(-\ell) - \int_{-\ell}^0 F
		\geq - \int_{-\ell}^0 F
		\quad \text{for each $\ell \in [0,+\infty)$}
	\end{equation*}
	by integration by parts,%
		\footnote{See e.g. Theorem~18.4 in \textcite{Billingsley1995}.}
	so $\int_{-\infty}^0 F < +\infty$ implies $I > -\infty$, while if $I > -\infty$ and $\lim_{\ell \to +\infty} \ell F(-\ell) = 0$ then $\int_{-\infty}^0 F < +\infty$. It remains only to show that $I > -\infty$ implies $\lim_{\ell \to +\infty} \ell F(-\ell) = 0$, or equivalently that $\limsup_{\ell \to +\infty} \ell F(-\ell) > 0$ implies $I = -\infty$. To that end, suppose that there exists a sequence $(\ell_n)_{n \in \N}$ in $\R$ along which $\ell_n \to +\infty$ and $\ell_n F(-\ell_n) \to \delta \neq 0$ as $n \to +\infty$. For all $n \in \N$ large enough that $-\ell_n < 0$, we have
	\begin{equation*}
		I
		= \int_{(-\infty,-\ell_n]} k F(\dd k)
		+ \int_{(-\ell_n,0]} k F(\dd k)
		\leq -\ell_n F(-\ell_n)
		+ \int_{(-\ell_n,0]} k F(\dd k) ,
	\end{equation*}
	so letting $n \to +\infty$ yields $I \leq \delta + I$, which since $\delta \neq 0$ and $I \leq 0$ implies $I=-\infty$, as desired.
\end{proof}
\renewcommand{\qedsymbol}{$\blacksquare$}

Let $m \coloneqq \int_{(0,+\infty)} k F(\dd k)$, and define $\chi : \R \to (-\infty,+\infty]$ by $\chi(\ell) \coloneqq \int \max\{\ell,k\} F(\dd k)$ for every $\ell \in \R$. By \Cref{lemma:chi} in \cref{app:lemma_chi}, $\chi<+\infty$ holds if and only if there exists an $\ell \in \R$ such that $\chi(\ell)<+\infty$, which holds if and only if $m<+\infty$; furthermore, if $\chi<+\infty$, then $\chi$ is absolutely continuous with $\chi' = F$ a.e. If $\chi<+\infty$, then for each $\ell \in \R$,
\begin{align*}
	\chi(\ell)
	= F(\ell) \ell
	+ \int_{(\ell,+\infty)} k F(\dd k)
	&=
	\begin{cases}
		m + F(\ell) \ell
		+ \int_{(\ell,0]} k F(\dd k)
		& \text{if $\ell \leq 0$}
		\\
		m + F(\ell) \ell
		- \int_{(0,\ell]} k F(\dd k)
		& \text{if $\ell > 0$}
	\end{cases}
	\\
	&= m - \int_\ell^0 F ,
\end{align*}
where the final equality holds by integration by parts.%
	\footnote{See e.g. Theorem~18.4 in \textcite{Billingsley1995}.}
Hence if $\chi<+\infty$, then $\chi$ is strictly increasing on $J \coloneqq \{ k \in \R : F(k) > 0 \}$, $\chi$ is unbounded above,%
	\footnote{Since $\lim_{k \nearrow +\infty} F(k) = 1$, there exists a $k \in \R$ such that $F(k) \geq 1/2$, so for each $\ell>\min\{0,k\}$, $\chi(\ell) = m + \int_{(0,\ell]} F \geq m + (\ell - k)/2 \to +\infty$ as $\ell \to +\infty$.}
and $\chi$ is bounded below if and only if $\int_{-\infty}^0 F < +\infty$, which by \Cref{observation:F_integrals} holds if and only if $\int_{(-\infty,0]} k F(\dd k) > -\infty$.

Suppose \ref{item:F_ident:oo} holds. Then $\chi \circ v = \alpha u + \beta < +\infty$ since $u$ is real-valued, so $m < +\infty$. If $\int_{(-\infty,0]} k F(\dd k) > -\infty$, then $u = ( \chi \circ v - \beta )$ is bounded below since $\chi$ is.

Suppose that \ref{item:F_ident:finite} holds. If $\int_{(-\infty,0]} k F(\dd k) > -\infty$, then both $\chi$ and $u$ are bounded below, and we let $\beta \coloneqq \inf u(X) - \inf \chi(J) - 1$; otherwise let $\beta \coloneqq 0$. Define $\phi : J \to \R$ by $\phi \coloneqq \chi + \beta$. Then $\phi$ is strictly increasing and unbounded above, and $u(X) \subseteq \phi(J)$ by definition of $\beta$, so we may define $v : X \to \R$ by $v \coloneqq \phi^{-1} \circ u$. Then $u = \phi \circ v$, and $\phi$ is absolutely continuous with $\phi' = F$ a.e., so by \Cref{proposition:fixedF}, there exist $\alpha>0$ and $\beta \in \R$ such that
\begin{equation*}
	\alpha u(x) + \beta
	= \int \max\{v(x),k\} F(\dd k)
	\quad \text{for every $x \in X$.}
	\pushQED{\qed}\qedhere
\end{equation*}

\section{Proof of \texorpdfstring{\Cref{theorem:mcs} (\cpageref{theorem:mcs})}{Theorem~\ref{theorem:mcs} (p. \pageref{theorem:mcs})}}
\label{app:mcs_pf}

Define $\chi : \co(v(X)) \to (-\infty,+\infty]$ by $\chi(\ell) \coloneqq \int \max\{\ell,k\} F(\dd k)$ for every $\ell \in \co(v(X))$. Since $u$ is real-valued, $\chi<+\infty$, so by \Cref{lemma:chi} in \cref{app:lemma_chi} and the (Lebesgue) fundamental theorem of calculus, $\chi(m)-\chi(\ell) = \int_\ell^m F$ for all $\ell<m$ in $\co(v(X))$. Thus $\chi$ is strictly increasing on $\cotwo(v(X))$ since $F>0$ on $(\inftwo v(X),+\infty)$, and if $F(\inf v(X))>0$ then $\chi$ is strictly increasing on $\co(v(X))$. Hence if $\cotwo(v(X))=\co(v(X))$ or $F(\inf v(X))>0$, then $\chi$ has a well-defined inverse $\chi^{-1} : \chi(\co(v(X))) \to \co(v(X))$; otherwise, $\chi$ has a well-defined inverse $\chi^{-1} : \chi(\cotwo(v(X))) \to \co(v(X))$, which is in fact a map $\chi(\co(v(X))) \to \co(v(X))$ since $\chi(\cotwo(v(X))) = \chi(\co(v(X)))$, as in this case $\chi$ is (constant and) equal to $\chi(\inf v(X))$ on $\co(v(X)) \setminus \cotwo(v(X))$ and $\inf v(X) \in \cotwo(v(X))$. Note that $\chi^{-1}$ is strictly increasing and onto $\cotwo(v(X))$, and that the image $\chi(\co(v(X)))$ is convex (since $\chi$ is continuous). Similarly, define $\widehat{\chi} : \co(\widehat{v}(X)) \to \R$ by $\widehat{\chi}(\ell) \coloneqq \int \max\{\ell,k\} \widehat{F}(\dd k)$ for every $\ell \in \co(\widehat{v}(X))$; it has the same properties.

\paragraph{`If' half of part~\ref{item:mcs:F}:}
Suppose that $\widehat{v}=v$ and that $\widehat{F}$ is better than $F$ in the reverse hazard rate order on $v(X) \setminus \{\sup v(X)\}$. Since $\widehat{F}$ and $F$ are both concentrated on $v(X) \cup \{-\infty\}$, $\widehat{F}$ is better than $F$ in the reverse hazard rate order on $\co(v(X)) \setminus \{\sup v(X)\}$. Observe that for all $k,\ell,m \in \co(v(X))$ such that $\chi(k) < \chi(\ell) < \chi(m)$, we have $k < \ell < m$ since $\chi$ is increasing, so that $\widehat{F}(\ell) > 0$ and 
\begin{equation*}
	\int_k^\ell F
	\geq \frac{F(\ell)}{\widehat{F}(\ell)} \int_k^\ell \widehat{F}
	\quad \text{and} \quad
	\int_\ell^m F
	\leq \frac{F(\ell)}{\widehat{F}(\ell)} \int_\ell^m \widehat{F} 
\end{equation*}
since $\widehat{F}$ is better than $F$ in the reverse hazard rate order on $\co(v(X)) \setminus \{\sup v(X)\}$, and thus
\begin{equation*}
	\frac{\widehat{\chi}(m) - \widehat{\chi}(\ell)}
	{\widehat{\chi}(\ell) - \widehat{\chi}(k)}
	= \frac{\int_\ell^m \widehat{F}}
	{\int_k^\ell \widehat{F}}
	\geq \frac{\int_\ell^m F}{\int_k^\ell F}
	= \frac{\chi(m)-\chi(\ell)}{\chi(\ell)-\chi(k)}.
\end{equation*}
Then by viewing $\co(v(X))$ as a set of alternatives and $\chi,\widehat{\chi} : \co(v(X)) \to \R$ as risk attitudes (and recalling that $\chi(\co(v(X))$ is convex), we see from \hyperref[theorem:pratt]{Pratt's theorem} (\cpageref{theorem:pratt}) that there exists a convex and strictly increasing map $\psi : \chi(\co(v(X))) \to \R$ such that $\widehat{\chi} = \psi \circ \chi$.

By hypothesis, there exist strictly increasing affine maps $\lambda,\widehat{\lambda} : \R \to \R$ such that $\lambda \circ u = \chi \circ v$ and $\widehat{\lambda} \circ \widehat{u} = \widehat{\chi} \circ \widehat{v}$. Hence $\widehat{\lambda} \circ \widehat{u} = \psi \circ \chi \circ v$ and $v = \chi^{-1} \circ \lambda \circ u$, so $\widehat{u} = \phi \circ u$ where $\phi \coloneqq \widehat{\lambda}^{-1} \circ \psi \circ \lambda$. The map $\phi : \co(u(X)) \to \R$ is strictly increasing since $\widehat{\lambda}$, $\psi$, and $\lambda$ are, and convex since $\psi$ is convex and $\widehat{\lambda}$ and $\lambda$ are affine. Hence $\widehat{u}$ is less risk-averse than $u$ by \hyperref[theorem:pratt]{Pratt's theorem}.

\paragraph{`Only if' half of part~\ref{item:mcs:F}:}
Suppose that $\widehat{v}=v$ and that $\widehat{u}$ is less risk-averse than $u$. Begin with a claim.

\begin{claim}
	\label{claim:mcs_F_onlyif}
	For any $k < \ell < m$ in $v(X)$, $\int_k^\ell F > 0 < \int_\ell^m F$ and
	\begin{equation*}
		\frac{\int_\ell^m \widehat{F}}{\int_\ell^m F}
		\geq \frac{\int_k^\ell \widehat{F}}{\int_k^\ell F} .
	\end{equation*}
\end{claim}

\begin{proof}[Proof of the {\hyperref[claim:mcs_F_onlyif]{claim}}]
	\renewcommand{\qedsymbol}{$\square$}
	We have $\chi(k) < \chi(\ell) < \chi(m)$ since $\chi$ is strictly increasing on $\cotwo(v(X)) \supseteq v(X)$, so
	\begin{gather*}
		\int_k^\ell F
		= \chi(\ell) - \chi(k)
		> 0
		< \chi(m) - \chi(\ell)
		= \int_\ell^m F ,
		\\
		\frac{\widehat{\chi}(m) - \widehat{\chi}(\ell)}{\widehat{\chi}(\ell) - \widehat{\chi}(k)}
		= \frac{\int_\ell^m \widehat{F}}{\int_k^\ell \widehat{F}}
		\quad \text{and} \quad
		\frac{\chi(m)-\chi(\ell)}{\chi(\ell)-\chi(k)}
		= \frac{\int_\ell^m F}{\int_k^\ell F} .
	\end{gather*}
	It remains only to show that
	\begin{equation*}
		\frac{\widehat{\chi}(m) - \widehat{\chi}(\ell)}{\widehat{\chi}(\ell) - \widehat{\chi}(k)}
		\geq \frac{\chi(m)-\chi(\ell)}{\chi(\ell)-\chi(k)} .
	\end{equation*}
	By viewing $v(X)$ as a set of alternatives and (the restrictions to $v(X)$ of) $\chi$ and $\widehat{\chi}$ as risk attitudes, we see from \hyperref[theorem:pratt]{Pratt's theorem} that it suffices to show that there exists an increasing convex function $\phi : \co(\chi(v(X))) \to \R$ that is strictly increasing on $\chi(v(X))$ and satisfies $\widehat{\chi} = \phi \circ \chi$ on $v(X)$.

	To that end, observe that by \hyperref[theorem:pratt]{Pratt's theorem}, since $\widehat{u}$ is less risk-averse than $u$, there exists an increasing convex function $\psi : \co(u(X)) \to \R$ that is strictly increasing on $u(X)$ and satisfies $\widehat{u} = \psi \circ u$. By hypothesis, there exist strictly increasing affine maps $\lambda,\widehat{\lambda} : \R \to \R$ such that $\lambda \circ u = \chi \circ v$ and $\widehat{\lambda} \circ \widehat{u} = \widehat{\chi} \circ v$. Hence $u = \lambda^{-1} \circ \chi \circ v$ and $\widehat{\lambda} \circ \psi \circ u = \widehat{\chi} \circ v$, so $\widehat{\chi} = \phi \circ \chi$ on $v(X)$ where $\phi \coloneqq \widehat{\lambda} \circ \psi \circ \lambda^{-1}$. The map $\phi : \co(\chi(v(X))) \to \R$ is increasing and convex since $\widehat{\lambda}$, $\psi$ and $\lambda$ are, and is strictly increasing on $\chi(v(X))$ since $\psi$ is strictly increasing on $u(X)$ and $\lambda$ and $\widehat{\lambda}$ are strictly increasing.
\end{proof}
\renewcommand{\qedsymbol}{$\blacksquare$}

Now, we must show that $\widehat{F}(k')F(\ell') \leq F(k')\widehat{F}(\ell')$ for all $k' < \ell'$ in $v(X) \setminus \{\sup v(X)\}$. If $\abs*{v(X)} < 3$, then this holds trivially. Assume for the remainder that $\abs*{v(X)} \geq 3$, and fix an arbitrary pair $k' < \ell'$ in $v(X) \setminus \{\sup v(X)\}$. Note that $v(X) \cap (\ell',+\infty)$ is non-empty.

Suppose first that $v(X) \cap (k',\ell') = \varnothing$. Then applying the \hyperref[claim:mcs_F_onlyif]{claim} with $k=k'$, $\ell=\ell'$ and $m \in v(X) \cap (\ell',+\infty)$ and letting $m \to \inf (v(X) \cap (\ell',+\infty))$ yields
\begin{equation*}
	\frac{\widehat{F}(\ell')}{F(\ell')}
	\geq \frac{\int_{k'}^{\ell'} \widehat{F}}{\int_{k'}^{\ell'} F}
	= \frac{(\ell'-k') \widehat{F}(k')}{(\ell'-k') F(k')}
	= \frac{\widehat{F}(k')}{F(k')} ,
\end{equation*}
where the inequality holds by the \hyperref[claim:mcs_F_onlyif]{claim} and the fact that $F$ and $\widehat{F}$ are concentrated on $v(X) \cup \{-\infty\}$, and the first equality holds since $F$ and $\widehat{F}$ are right-continuous and concentrated on $v(X) \cup \{-\infty\}$.

Suppose instead that $v(X) \cap (k',\ell') \neq \varnothing$. For any $\ell \in v(X) \cap (k',\ell')$ and $m \in v(X) \cap (\ell',+\infty)$, applying the \hyperref[claim:mcs_F_onlyif]{claim} twice yields
\begin{equation*}
	\frac{\int_{\ell'}^m \widehat{F}}{\int_{\ell'}^m F} \geq \frac{\int_\ell^{\ell'} \widehat{F}}{\int_\ell^{\ell'} F} \geq \frac{\int_{k'}^\ell \widehat{F}}{\int_{k'}^\ell F} .
\end{equation*}
Letting $m \to \inf (v(X) \cap (\ell',+\infty))$ and $\ell \to \inf (v(X) \cap (k',+\infty)) \eqqcolon k^\star$ yields  
\begin{equation*}
	\widehat{F}(\ell') \geq F(\ell')\liminf_{\ell \searrow k^\star} \frac{\int_\ell^{\ell'} \widehat{F}}{\int_\ell^{\ell'} F} \geq F(\ell')\frac{\widehat{F}(k')}{F(k')} 
\end{equation*}
since $F$ and $\widehat{F}$ are concentrated on $v(X) \cup \{-\infty\}$.

\bigskip

It remains to prove part~\ref{item:mcs:v}. For both halves (`if' and `only if'), we will use the following simple observation.

\begin{observation}
	\label{observation:mcs_v_psi}
	If $F \circ v = \widehat{F} \circ \widehat{v}$, and there exists an increasing function $\psi : v(X) \to \R$ such that $\widehat{v} = \psi \circ v$, then 
	\begin{equation*}
		\left(\widehat{\chi} \circ \psi\right)(m)
		= \int \max\{\psi(m),\psi(k)\} F(\dd k)
		= \int \psi(\max\{m,k\}) F(\dd k) .
	\end{equation*}
\end{observation}

\begin{proof}[Proof of \Cref{observation:mcs_v_psi}]
	\renewcommand{\qedsymbol}{$\square$}
	Since $F \circ v = \widehat{F} \circ \widehat{v}$ and $\widehat{v} = \psi \circ v$, we have $F = \widehat{F} \circ \psi$ on $v(X)$. Since $F$ is concentrated on $v(X) \cup \{-\infty\}$ and $\widehat{F}$ is concentrated on $\widehat{v}(X) \cup \{-\infty\}$, it follows that for each $m \in \co(v(X))$,
	\begin{multline*}
		\left(\widehat{\chi} \circ \psi\right)(m)
		= \int \max\left\{\psi(m),\widehat{k}\right\} \widehat{F}\left(\dd \widehat{k}\right)
		= \int \max\{\psi(m),\psi(k)\} F(\dd k)
		\\
		= \int \psi(\max\{m,k\}) F(\dd k) ,
	\end{multline*}
	where the second equality holds by change of variable,%
		\footnote{Licensed by Theorem~16.13 in \textcite{Billingsley1995}, which is applicable since the integrand $k \mapsto \max\{\psi(m),k\}$ is bounded below (by $\psi(m)$).}
	and the final equality holds since $\psi$ is increasing.
\end{proof}
\renewcommand{\qedsymbol}{$\blacksquare$}

\paragraph{`If' half of part~\ref{item:mcs:v}:}
Suppose that $F \circ v = \widehat{F} \circ \widehat{v}$ and that $\widehat{v}$ is less risk-averse than $v$. Then by \hyperref[theorem:pratt]{Pratt's theorem} and \Cref{lemma:greatest_phi} in \cref{app:pratt_phis}, there exists an increasing convex map $\psi : \co(v(X)) \to \R$ that is strictly increasing on $\cotwo(v(X))$ and satisfies $\widehat{v} = \psi \circ v$. By hypothesis, there exist strictly increasing affine maps $\lambda,\widehat{\lambda} : \R \to \R$ such that $\lambda \circ u = \chi \circ v$ and $\widehat{\lambda} \circ \widehat{u} = \widehat{\chi} \circ \widehat{v}$. Hence $\widehat{\lambda} \circ \widehat{u} = \widehat{\chi} \circ \psi \circ v$ and $v = \chi^{-1} \circ \lambda \circ u$, so $\widehat{u} = \phi \circ u$ where $\phi \coloneqq \widehat{\lambda}^{-1} \circ \widehat{\chi} \circ \psi \circ \chi^{-1} \circ \lambda$. By \hyperref[theorem:pratt]{Pratt's theorem}, it suffices to show that the map $\phi : \co(u(X)) \to \R$ is increasing and convex, and strictly increasing on $\cotwo(u(X))$. That $\phi$ is increasing follows from the fact that $\widehat{\lambda}$, $\widehat{\chi}$, $\psi$, $\chi$ and $\lambda$ are. $\phi$ is strictly increasing on $\cotwo(u(X))$ since $\lambda$ and $\widehat{\lambda}$ are strictly increasing, $\chi^{-1}$ is strictly increasing and onto $\cotwo(v(X))$, $\psi$ is strictly increasing on $\cotwo(v(X))$ and has $\psi(\cotwo(v(X))) \subseteq \cotwo(\psi(v(X))) = \cotwo(\widehat{v}(X))$, and $\widehat{\chi}$ is strictly increasing on $\cotwo(\widehat{v}(X))$. To show that $\phi$ is convex, it is enough (since $\widehat{\lambda}$ and $\lambda$ are affine) to show that $\widehat{\chi} \circ \psi \circ \chi^{-1}$ is convex. To that end, we shall exhibit a convex function $f : \chi(\co(v(X))) \to \R$ such that $\widehat{\chi} \circ \psi = f \circ \chi$; this suffices since $\chi \circ \chi^{-1}$ is the identity, i.e. $\chi(\chi^{-1}(k))=k$ for every $k \in \chi(\co(v(X)))$.

We claim that $\psi$ may without loss of generality be assumed to be continuous. To see why, write $m \coloneqq \sup v(X)$, and observe that (being convex and increasing,) $\psi$ is automatically continuous on $\co(v(X)) \setminus \{m\}$. Hence $\psi = \psi_2 \circ \psi_1$, where $\psi_1 : \co(v(X)) \to \R$ satisfies $\psi_1 \coloneqq \psi$ on $\co(v(X)) \setminus \{m\}$, $\psi_2 : \psi_1(\co(v(X))) \to \R$ is equal to the identity on $\psi_1(\co(v(X))) \setminus \{\sup \psi_1(v(X))\}$, and if $m \in v(X)$, then $\psi_1(m) \coloneqq \sup \psi\bigl( \co(v(X)) \setminus \{m\} \bigr)$ and $\psi_2(\psi_1(m)) \coloneqq \psi(m)$. Evidently $\psi_1$ and $\psi_2$ are convex and strictly increasing, and $\psi_1$ is continuous; furthermore, $\psi = \psi_1$ if and only if $\psi$ is continuous. Let $F_1$ be the extended CDF concentrated on $\psi_1(v(X)) \cup \{-\infty\}$ such that $F_1>0$ on $(\inf \psi_1(v(X)),+\infty)$ and $F_1 \circ \psi_1 \circ v = F \circ v$, and let $\chi_1 : \psi_1(\co(v(X))) \to \R$ be given by $\chi_1(\ell) \coloneqq \int \max\{\ell,k\} F_1(\dd k)$ for every $\ell \in \psi_1(\co(v(X)))$. If $\psi$ is not continuous, then $m \in v(X)$ and $\widehat{\chi}\circ \psi = f_2 \circ \chi_1 \circ \psi_1$, where $f_2$ is the convex function $\chi_1(\psi_1(\co(v(X)))) \to \R$ given by
\begin{equation*}
	f_2(\ell)
	\coloneqq \ell + \bigl[ \psi\left(m\right)-\psi_1\left(m\right) \bigr] \int_{\left\{m\right\}} \dd \widehat{F}
\end{equation*}
for each $\ell \in \chi_1(\psi_1(\co(v(X)))) \setminus \{\psi_1\left(m\right)\}$ and $f_2\left(\psi_1\left(m\right)\right) \coloneqq \psi\left(m\right)$. Hence if $\psi$ is not continuous, then to show that there exists a convex function $f : \chi(\co(v(X))) \to \R$ such that $\widehat{\chi} \circ \psi = f \circ \chi$, it suffices to find a convex $f_1 : \chi(\co(v(X))) \to \R$ such that $\chi_1 \circ \psi_1 = f_1 \circ \chi$, because then we may set $f \coloneqq f_2 \circ f_1$. Since if $\psi$ is continuous then $\psi_1=\psi$ and $\chi_1=\widehat{\chi}$, it follows that it is without loss of generality to assume that $\psi$ is continuous.

We assume for the remainder that $\psi$ is continuous. It remains to show that there exists a convex function $f : \chi(\co(v(X))) \to \R$ such that $\widehat{\chi} \circ \psi = f \circ \chi$. By viewing $Y \coloneqq \cotwo(v(X))$ as a set of alternatives and $(\widehat{\chi} \circ \psi)|_Y$ and $\chi|_Y$ (the restrictions to $Y$ of $\widehat{\chi} \circ \psi$ and $\chi$) as risk attitudes, we see from \hyperref[theorem:pratt]{Pratt's theorem} (noting that $\chi(\co(v(X))) = \co(\chi(Y))$) that it suffices to show that $(\widehat{\chi} \circ \psi)|_Y$ is less risk-averse than $\chi|_Y$. So fix any $\ell \in \cotwo(v(X))$ and any CDF $H$ whose support is finite and contained in $\cotwo(v(X))$; we must show that if $\int \chi \dd H \geq \mathrel{(>)} \chi(\ell)$, then $\int \bigl( \widehat{\chi} \circ \psi \bigr) \dd H \geq \mathrel{(>)} \widehat{\chi}(\psi(\ell))$. If $\ell = \inf v(X)$, then this holds trivially.%
	\footnote{If $H(\inf v(X))=1$ then $\int \chi \dd H = \chi(\inf v(X))$ and $\int \bigl( \widehat{\chi} \circ \psi \bigr) \dd H = \widehat{\chi}(\psi(\inf v(X)))$, while if $H(\inf v(X))<1$ then $\int \bigl( \widehat{\chi} \circ \psi \bigr) \dd H > \widehat{\chi}(\psi(\inf v(X)))$ since $\widehat{\chi}$ is strictly increasing on $\cotwo(\widehat{v}(X)) = \cotwo(\psi(v(X))) \supseteq \psi(\cotwo(v(X)))$ and $\psi$ is strictly increasing on $\cotwo(v(X))$.}
Suppose for the remainder that $\ell > \inf v(X)$.

Let $g : \co(v(X)) \times \co(v(X)) \to \co(v(X))$ and $J_\star, J^\star : \R^2 \to [0,1]$ be given by $g(m,k) \coloneqq \max\{m,k\}$, $J_\star(m,k) \coloneqq \1_{[\ell,+\infty)}(m) \times F(k)$ and $J^\star(m,k) \coloneqq H(m) \times F(k)$ for all $(m,k) \in \R^2$.
Define $H_\star,H^\star : \R \to [0,1]$ by $H_\star(k) \coloneqq J_\star(k,k)$ and $H^\star(k) \coloneqq J^\star(k,k)$ for each $k \in \R$, and let $B : \R \to [-1,1]$ be given by $B \coloneqq H_\star-H^\star$. Note that $B$ vanishes outside $\co(v(X))$.
For any increasing and continuous function $f : \co(v(X)) \to \R$ such that $f \circ g$ is both $J_\star$- and $J^\star$-integrable,
\begin{multline}
	\int \left[ \int f(\max\{m,k\}) F(\dd k) \right] H(\dd m)
	- \int f(\max\{\ell,k\}) F(\dd k)
	\\
	= \int ( f \circ g ) \dd J^\star
	- \int ( f \circ g ) \dd J_\star
	= - \int f \dd B
	= \int f' B ,
	\label{eq:mcs_v_f}
\end{multline}
where the first equality holds by Fubini's theorem,%
	\footnote{See e.g. Theorem~18.3 in \textcite{Billingsley1995}, which is applicable since $f \circ g$ is both $J_\star$- and $J^\star$-integrable.}
the second by change of variable,%
	\footnote{Licensed by Theorem~16.13 in \textcite{Billingsley1995}, which is applicable since $f \circ g$ is both $J_\star$- and $J^\star$-integrable.}
and the third by integration by parts.%
	\footnote{For any $k<m$, applying Theorem~18.4 in \textcite{Billingsley1995} separately to $H_\star$ and $H^\star$ yields $\int_{(k,m]} f' B = \int_{(k,m]} H_\star \dd f - \int_{(k,m]} H^\star \dd f = f(m) B(m) - f(k) B(k) - \int_{(k,m]} f \dd B$. Since $H$ has finite support, we have $B(k)=B(m)=0$ for all sufficiently small $k \in \R$ and sufficiently large $m \in \R$, so letting $k \searrow -\infty$ and $m \nearrow +\infty$ yields $\int f' B = - \int f \dd B$.}
In particular, since $g$ and $\psi \circ g$ are $J_\star$- and $J^\star$-integrable ($g$ since $\chi<+\infty$ and $H$ has finite support, and $\psi \circ g$ by \Cref{observation:mcs_v_psi} since $\widehat{\chi}<+\infty$ and $H$ has finite support), \eqref{eq:mcs_v_f} holds when $f$ is the identity $k \mapsto k$ and when $f = \psi$.

Since $\psi$ is convex, its right-hand derivative $\psi^+ : \co(v(X)) \setminus \{ \sup v(X) \} \to \R$ exists and is increasing. In case $\sup v(X) \in v(X)$, extend $\psi^+$ to $\co(v(X))$ by letting $\psi^+( \sup v(X) ) \coloneqq \lim_{k \nearrow \sup v(X)} \psi^+(k)$. Since $\psi$ is \emph{strictly} increasing on $\cotwo(v(X))$ and $\ell \in \cotwo(v(X)) \setminus \{\inf v(X)\}$, we have $\psi^+(\ell)>0$.%
	\footnote{If $\psi(\ell)=0$ then $\psi^+=0$ on $(\inf v(X), \ell )$, in which case ($\psi$ being absolutely continuous on $\interior(\co(v(X)))$, since it is convex) we would have for any $k<m$ in $(\inftwo v(X), \ell )$ that $\psi(m)-\psi(k) = \int_k^m \psi^+ = 0$ by the (Lebesgue) fundamental theorem of calculus, a contradiction with the fact that $\psi$ is strictly increasing on $\cotwo(v(X))$.}
Hence if $\int \chi \dd H \geq \mathrel{(>)} \chi(\ell)$, then
\begin{multline*}
	\int \left( \widehat{\chi} \circ \psi \right) \dd H
	- \widehat{\chi}(\psi(\ell))
	\\
	\begin{aligned}			
		&= \int \left[ \int \psi(\max\{m,k\}) F(\dd k) \right] H(\dd m)
		- \int \psi(\max\{\ell,k\}) F(\dd k)
		\\
		&= \int \psi^+ B 
		= \int_{-\infty}^\ell \psi^+ B
		+ \int_\ell^{+\infty} \psi^+ B
		\\
		&\geq \psi^+(\ell) \int_{-\infty}^\ell B
		+ \psi^+(\ell) \int_\ell^{+\infty} B
		= \psi^+(\ell)
		\int B
		\\
		&= \psi^+(\ell)
		\left(
		\int \chi \dd H - \chi(\ell)
		\right)
		\\
		&\geq \mathrel{(>)} 0 ,
	\end{aligned}
\end{multline*}
where the first equality holds by \Cref{observation:mcs_v_psi}, the second equality holds by applying \eqref{eq:mcs_v_f} with $f=\psi$, the first inequality holds since $\psi^+$ is increasing and $B(k) = -F(k)H(k) \leq 0 \leq F(m)[1-H(m)] = B(m)$ for any $k,m \in \R$ such that $k < \ell \leq m$, the final equality holds by \eqref{eq:mcs_v_f} with $f$ the identity $k \mapsto k$, and the final inequality holds since $\psi^+(\ell)>0$ and $\int \chi \dd H \geq \mathrel{(>)} \chi(\ell)$.

\paragraph{`Only if' half of part~\ref{item:mcs:v}:}
suppose that $F \circ v = \widehat{F} \circ \widehat{v}$ and that $\widehat{u}$ is less risk-averse than $u$. Then since $F>0$ on $(\inftwo v(X),+\infty)$ and $\widehat{F}>0$ on $\left(\inftwo \widehat{v}(X),+\infty\right)$, there exists a strictly increasing function $\psi : v(X) \to \R$ such that $\widehat{v} = \psi \circ v$.

\begin{claim}
	\label{claim:slopes}
	Let $(k_n,\ell_n,k'_n,\ell'_n)_{n \in \N}$ be a sequence in $v(X)^4$ that converges to $(k_\star,\ell_\star,k'_\star,\ell'_\star) \in \co(v(X))^4$ as $n \to +\infty$, and suppose that
	
	\begin{enumerate}[label=(\roman*)]
		
		\item \label{item:ranked} $k_n < \ell_n \leq k'_n < \ell'_n$ for each $n \in \N$, 
		
		\item \label{item:tight} $v(X) \cap \left(k_\star,\ell_\star\right) = \varnothing = v(X) \cap \left(k'_\star,\ell'_\star\right)$, and 
		
		\item \label{item:conv} The sequences
		\begin{equation*}
			(\delta_n)_{n \in \N} \coloneqq \left( \tfrac{\psi(\ell_n)-\psi(k_n)}{\ell_n-k_n}\right)_{n \in \N}
			\quad \text{and} \quad
			(\delta'_n)_{n \in \N} \coloneqq \left(\tfrac{\psi(\ell'_n)-\psi(k'_n)}{\ell'_n-k'_n} \right)_{n \in \N}
		\end{equation*}
		are monotone.

	\end{enumerate}
	Then $\lim_{n \to +\infty} \delta_n \leq \lim_{n \to +\infty} \delta_n'$.
\end{claim}

\begin{proof}[Proof of the {\hyperref[claim:slopes]{claim}}]
	\renewcommand{\qedsymbol}{$\square$}
	Since $\widehat{u}$ is less risk-averse than $u$, there exists an increasing convex function $\psi^\star : \co(u(X)) \to \R$ that is strictly increasing on $u(X)$ and satisfies $\widehat{u} = \psi^\star \circ u$. By hypothesis, there exist strictly increasing affine maps $\lambda,\widehat{\lambda} : \R \to \R$ such that $\lambda \circ u = \chi \circ v$ and $\widehat{\lambda} \circ \widehat{u} = \widehat{\chi} \circ \widehat{v}$. Hence $u = \lambda^{-1} \circ \chi \circ v$ and $\widehat{\chi} \circ \psi \circ v = \widehat{\lambda} \circ \psi^\star \circ u$, so $\widehat{\chi} \circ \psi = \phi \circ \chi$ on $v(X)$ where $\phi \coloneqq \widehat{\lambda} \circ \psi^\star \circ \lambda^{-1}$. The map $\phi : \co(\chi(v(X))) \to \R$ is convex since $\psi^\star$ is convex and $\widehat{\lambda}$ and $\lambda$ are affine, and is strictly increasing on $\chi(v(X))$ since $\chi$ is increasing, $\psi$ is strictly increasing, and $\widehat{\chi}$ is strictly increasing on $\cotwo(\widehat{v}(X)) \supseteq \widehat{v}(X) = \psi(v(X))$. By viewing $v(X)$ as a set of alternatives and (the restrictions to $v(X)$ of) $\chi$ and $\widehat{\chi}$ as risk attitudes, and noting that 
	\begin{equation*}
		\widehat{\chi}(\psi(k_n))
		< \widehat{\chi}(\psi(\ell_n))
		\leq \widehat{\chi}(\psi(k_n'))
		< \widehat{\chi}(\psi(\ell_n'))
		\quad \text{for each $n \in \N$}
	\end{equation*}
	since \ref{item:ranked} holds, $\psi$ is strictly increasing, and $\widehat{\chi}$ is strictly increasing on $\cotwo(\widehat{v}(X)) \supseteq \widehat{v}(X) = \psi(v(X))$, we see from \hyperref[theorem:pratt]{Pratt's theorem} that
	\begin{equation}
		\frac{\widehat{\chi}(\psi(\ell_n))-\widehat{\chi}(\psi(k_n))}{\chi(\ell_n)-\chi(k_n)}
		\leq \frac{\widehat{\chi}(\psi(\ell_n'))-\widehat{\chi}(\psi(k_n'))}{\chi(\ell_n')-\chi(k_n')}
		\quad \text{for every $n \in \N$.\footnotemark}
		\label{eq:mcs_v_chirank}
	\end{equation}%
		\footnotetext{If $\ell_n<k_n'$, then we apply \hyperref[theorem:pratt]{Pratt's theorem} twice:
		\begin{equation*}
			\frac{\widehat{\chi}(\psi(\ell_n))-\widehat{\chi}(\psi(k_n))}{\chi(\ell_n)-\chi(k_n)}
			\leq \frac{\widehat{\chi}(\psi(k_n'))-\widehat{\chi}(\psi(\ell_n))}{\chi(k_n')-\chi(\ell_n)}
			\\
			\leq \frac{\widehat{\chi}(\psi(\ell_n'))-\widehat{\chi}(\psi(k_n'))}{\chi(\ell_n')-\chi(k_n')} .
		\end{equation*}}%

	For any $k<\ell$ in $v(X)$,
	\begin{equation*}
		\Delta(k,\ell) \coloneqq \frac{F(k) + \int_{(k,\ell]} \frac{\psi(\ell)-\psi(m)}{\psi(\ell)-\psi(k)}F(\dd m)}{F(k) + \int_{(k,\ell]} \frac{\ell-m}{\ell-k}F(\dd m)} 
	\end{equation*}
	is well-defined since $F > 0$ on $(\inftwo v(X),+\infty)$. By inspection,
	\begin{align}
		\frac{\widehat{\chi}(\psi(\ell))-\widehat{\chi}(\psi(k))}{\chi(\ell)-\chi(k)}
		&= \frac{\int \bigl( \max\{\psi(\ell),\psi(m)\} - \max\{\psi(k),\psi(m)\} \bigr) F(\dd m)}{\int \bigl( \max\{\ell,m\} - \max\{k,m\} \bigr) F(\dd m)} 
		\nonumber
		\\
		&= \frac{\psi(\ell)-\psi(k)}{\ell-k} \Delta(k,\ell) 
		\quad \text{for any $k<\ell$ in $v(X)$,}
		\label{eq:mcs_v_Delta}
	\end{align}
	where the first equality holds by \Cref{observation:mcs_v_psi}. Since $\psi$ is increasing, \ref{item:tight} holds, and $F$ is concentrated on $v(X) \cup \{-\infty\}$, we have
	\begin{equation}
		\lim_{n \to +\infty} \Delta(k_n,\ell_n)
		= 1
		= \lim_{n \to +\infty} \Delta(k'_n,\ell'_n)
		\label{eq:mcs_v_limit} 
	\end{equation}
	by the bounded convergence theorem.

	The limits $\lim_{n \to +\infty} \delta_n$ and $\lim_{n \to +\infty} \delta_n'$ exist in $\R \cup \{-\infty,+\infty\}$ by \ref{item:conv}, and they satisfy
	\begin{multline*}
		\lim_{n \to +\infty} \delta_n
		= \left.
		\lim_{n \to +\infty} \frac{\widehat{\chi}(\psi(\ell_n))-\widehat{\chi}(\psi(k_n))}{\chi(\ell_n)-\chi(k_n)} 
		\middle/
		\lim_{n \to +\infty} \Delta(k_n,\ell_n)
		\right.
		\\
		= \lim_{n \to +\infty} \frac{\widehat{\chi}(\psi(\ell_n))-\widehat{\chi}(\psi(k_n))}{\chi(\ell_n)-\chi(k_n)} 
		\leq \lim_{n \to +\infty} \frac{\widehat{\chi}(\psi(\ell'_n))-\widehat{\chi}(\psi(k'_n))}{\chi(\ell'_n)-\chi(k'_n)}
		\\
		= \left.
		\lim_{n \to +\infty} \frac{\widehat{\chi}(\psi(\ell'_n))-\widehat{\chi}(\psi(k'_n))}{\chi(\ell'_n)-\chi(k'_n)}
		\middle/
		\lim_{n \to +\infty} \Delta(k'_n,\ell'_n)
		\right.
		= \lim_{n \to +\infty}\delta_n' ,
	\end{multline*}
	where the first and last equalities hold by \eqref{eq:mcs_v_Delta}, the middle equalities hold by \eqref{eq:mcs_v_limit}, and the inequality holds by \eqref{eq:mcs_v_chirank}.
\end{proof}
\renewcommand{\qedsymbol}{$\blacksquare$}

The \hyperref[claim:slopes]{claim} implies that
\begin{equation}
	\frac{\psi(\ell')-\psi(k')}{\ell'-k'}
	\geq \frac{\psi(\ell)-\psi(k)}{\ell-k}
	\quad \text{for all $k < \ell \leq k' < \ell'$ in $v(X)$.}
	\label{eq:mcs_v_ineq}
\end{equation}
To see why, suppose toward a contradiction that
\begin{equation*}
	\delta
	\coloneqq \frac{\psi(\ell)-\psi(k)}{\ell-k}
	> \frac{\psi(\ell')-\psi(k')}{\ell'-k'}
	\eqqcolon \delta' 
	\quad \text{for some $k < \ell \leq k' < \ell'$ in $v(X)$.}
\end{equation*}
Let $(k_1,\ell_1) \coloneqq (k,\ell)$, and for each $n \in \N$, if $v(X) \cap (k_n,\ell_n) = \varnothing$ then let $(k_{n+1},\ell_{n+1}) \coloneqq (k_n,\ell_n)$, and otherwise fix an $m_n \in v(X) \cap (k_n,\ell_n)$ and choose $(k_{n+1},\ell_{n+1}) \in \{(k_n,m_n),(m_n,\ell_n)\}$ to maximise $\delta_{n+1}$. Similarly, let $(k_1',\ell_1') \coloneqq (k',\ell')$, and for each $n \in \N$, if $v(X) \cap (k_n',\ell_n') = \varnothing$ then let $(k_{n+1}',\ell_{n+1}') \coloneqq (k_n',\ell_n')$, and otherwise fix an $m_n' \in v(X) \cap (k_n',\ell_n')$ and choose $(k_{n+1}',\ell_{n+1}') \in \{(k_n',m_n'),(m_n',\ell_n')\}$ to \emph{minimise} $\delta_{n+1}'$. The sequence $(k_n,k'_n)_{n \in \N}$ increasing and bounded above by $(\ell,\ell')$, so converges to some $(k_\star,k_\star') \in \co(v(X))^2$ as $n \to +\infty$. Similarly, $(\ell_n,\ell_n')_{n \in \N}$ is decreasing and bounded below by $(k,k')$, so converges to some $(\ell_\star,\ell_\star') \in \co(v(X))^2$. Evidently \ref{item:ranked} is satisfied, and by appropriately choosing $(m_n,m_n')_{n \in \N}$, we may ensure that \ref{item:tight} holds. Since $(\delta_n)_{n \in \N}$ is increasing and $(\delta'_n)_{n \in \N}$ decreasing, \ref{item:conv} holds, and $\delta_n \geq \delta_1 = \delta > \delta' = \delta_1' \geq \delta_n'$ for every $n \in \N$, whence $\lim_{n \to +\infty} \delta_n > (\delta+\delta')/2 > \lim_{n \to +\infty} \delta_n'$, a contradiction with the \hyperref[claim:slopes]{claim}.

Since $\psi$ is strictly increasing, it holds for any $x,y \in X$ that $\widehat{v}(x) \geq \mathrel{(>)} \widehat{v}(y)$ implies $v(x) \geq \mathrel{(>)} v(y)$. Hence \eqref{eq:mcs_v_ineq} implies that for any alternatives $x,y,z \in X$, if $\widehat{v}(x)<\widehat{v}(y)<\widehat{v}(z)$, then
\begin{equation*}
	\frac{\widehat{v}(z)-\widehat{v}(y)}{\widehat{v}(y)-\widehat{v}(x)} 
	\geq \frac{v(z)-v(y)}{v(y)-v(x)} .
\end{equation*}
Thus $\widehat{v}$ is less risk-averse than $v$ by \hyperref[theorem:pratt]{Pratt's theorem}.
\qed

\section{Proof of \texorpdfstring{\Cref{theorem:special} (\cpageref{theorem:special})}{Theorem~\ref{theorem:special} (p.~\pageref{theorem:special})}}
\label{app:special_pf}

We shall prove a three-way equivalence between properties~\ref{item:special:lra}, \ref{item:special:oo}, and

\begin{enumerate}[label=(\alph*)]

	\setcounter{enumi}{2}

	\item \label{item:special:distn}
	There exist $\lambda \in (0,1]$, an extended CDF $G$ concentrated on $X \cup \{-\infty\}$ with $G>0$ on $(\inf X,+\infty)$, and a CDF $H$ concentrated on $X$ such that $G_x = \lambda G \1_{[x,+\infty)} + (1-\lambda)H$ for every $x \in X$.
		
\end{enumerate}
We first show that \ref{item:special:distn} implies \ref{item:special:oo} and that \ref{item:special:oo} implies \ref{item:special:lra}, and then prove that \ref{item:special:lra} implies \ref{item:special:distn}. The first two are easy, while the last is non-trivial.

\begin{proof}[Proof that \ref{item:special:distn} implies \ref{item:special:oo}]
	Suppose that \ref{item:special:distn} holds, and fix a bounded and strictly increasing function $v : X \to \R$. Define $\alpha>0$ and $\beta \in \R$ by $\alpha \coloneqq 1/\lambda$ and $\beta \coloneqq (1-1/\lambda) \int v \dd H$. Then for each $x \in X$, we have
	\begin{multline*}
		\alpha \int v \dd G_x + \beta
		= \int v \dd [\1_{[x,+\infty)}G]
		= G(x) v(x) + \int_{(x,+\infty)} v \dd G
		\\
		= \int v(\max\{x,y\}) G(\dd y)
		= \int \max\{v(x),v(y)\} G(\dd y) ,
	\end{multline*}
	where the last equality holds since $v$ is increasing.
\end{proof}

\begin{proof}[Proof that \ref{item:special:oo} implies \ref{item:special:lra}]
	Suppose that \ref{item:special:oo} holds, and fix a bounded and strictly increasing function $v : X \to \R$. By hypothesis, there exist $\alpha>0$, $\beta \in \R$, and an extended CDF $G$ concentrated on $X \cup \{-\infty\}$ with $G>0$ on $(\inf X,+\infty)$ such that
	\begin{equation*}
		\alpha \int v \dd G_x + \beta
		= \int \max\{ v(x), v(y) \} G(\dd y)
		\quad \text{for every $x \in X$.}
	\end{equation*}
	Let $F$ be the (unique) extended CDF satisfying $F \coloneqq G \circ v^{-1}$. Then
	\begin{equation*}
		\alpha \int v \dd G_x + \beta
		= \int \max\{ v(x), k \} F(\dd k)
		\quad \text{for every $x \in X$}
	\end{equation*}
	by change of variable,%
		\footnote{Licensed by Theorem~16.13 in \textcite{Billingsley1995}, which is applicable since for each $x \in X$, the integrand $k \mapsto \max\{v(x),k\}$ is bounded below (by $v(x)$).}
	and $F>0$ on $(\inf v(X),+\infty) \supseteq (\inftwo v(X),+\infty)$. Hence $x \mapsto \int v \dd G_x$ is less risk-averse than $v$ by \Cref{theorem:ident}.
\end{proof}

\begin{proof}[Proof that \ref{item:special:lra} implies \ref{item:special:distn}]
	Suppose that \ref{item:special:lra} holds. For any $x < y$ in $X$ and any bounded and strictly increasing function $v : X \to \R$, we have $v(x) < v(y)$, so $\int v \dd G_x < \int v \dd G_y$ by \hyperref[theorem:pratt]{Pratt's theorem} (\cpageref{theorem:pratt}). It follows that $x \mapsto G_x$ is strictly increasing with respect to first-order stochastic dominance, i.e. $G_x \geq G_y \neq G_x$ for any $x<y$ in $X$.%
		\footnote{Suppose toward a contradiction that $G_x(z) < G_y(z)$ for some $x<y$ in $X$ and some $z \in \R$. Since $G_x$ and $G_y$ are concentrated on $X$, we may choose $z \in X$. Let $w : X \to \R$ be bounded and strictly increasing, and for each $\eps>0$, define $v_\eps : X \to \R$ by $v_\eps \coloneqq \eps w + \1_{(z,+\infty)}$. Then $\int v_\eps G_x > \int v_\eps G_y$ for all sufficiently small $\eps>0$ by the bounded convergence theorem, even though $v_\eps$ is bounded and strictly increasing---a contradiction.}

	We shall make repeated use of the fact that for any CDF $G_\star$ and any increasing, bounded and absolutely continuous function $w : \R \to \R$,
	\begin{equation}
		\int w \dd G_\star
		= \inf w(\R) + \int w' (1-G_\star)
		\label{eq:mean_ibp}
	\end{equation}
	by integration by parts.%
		\footnote{See e.g. Theorem~18.4 in \textcite{Billingsley1995}.}

	\begin{namedthm}[Claim 1.]
		\label{claim:LR}
		For all $x < y$ in $X$, we have $G_x = G_y$ on $\R \setminus [x,y)$.
	\end{namedthm}

	\begin{proof}[Proof of {\hyperref[claim:LR]{claim~1}}]
		\renewcommand{\qedsymbol}{$\square$}
		Suppose toward a contradiction that there exist $x<y$ in $X$ such that $G_x \neq G_y$ on $\R \setminus [x,y)$. Then we may find a $\delta_1 \in \R$ large enough that $\int_{(-\delta_1,x) \cup [y,\delta_1)} (G_x-G_y) > 0$. Since $y<\sup X$ (as $y \in X$ and $\sup X \notin X$), we may choose a $\delta_2 \in \R$ large enough that $X \cap (y,\delta_2)$ is non-empty. Let $\delta \coloneqq \max\{\delta_1,\delta_2\}$, fix a $z \in X \cap (y,\delta)$, and write $A \coloneqq (-\delta,x) \cup [y,\delta)$.

		Fix a Lebesgue-integrable function $\pi : \R \to (0,+\infty)$.%
			\footnote{For example, $\pi(x') \coloneqq \exp(-\abs*{x'})$ for every $x' \in \R$.}
		For each $\eps>0$, define $w_\eps : \R \to \R$ by $w_\eps(x') \coloneqq \int_0^{x'} (\1_A + \eps \pi \1_{\R \setminus A})$ for each $x' \in \R$, and note that $w_\eps$ is strictly increasing, bounded, and absolutely continuous, with (a.e.-defined) derivative $w_\eps'=1$ on $A$ and $w_\eps'=\eps \pi$ on $\R \setminus A$.%
			\footnote{$w_\eps$ is absolutely continuous by the (Lebesgue) fundamental theorem of calculus, and (hence) strictly increasing since $w_\eps' > 0$.}
		For each $\eps>0$, let $v_\eps \coloneqq w_\eps|_X$ be the restriction of $w_\eps$ to $X$. For any sufficiently small $\eps>0$,
		\begin{multline*}
			\frac{\int v_\eps \dd(G_z-G_y)}{\int v_\eps \dd(G_y-G_x)} 
			= \frac{\int_{A} (G_y-G_z) + \eps \int_{\R \setminus A} \pi (G_y-G_z)}{\int_{A} (G_x-G_y) + \eps \int_{\R \setminus A} \pi (G_x-G_y)} 
			\\
			< \frac{z-y}{\eps \int_x^y \pi}
			= \frac{v_\eps(z)-v_\eps(y)}{v_\eps(y)-v_\eps(x)} , 
		\end{multline*}
		where the first equality holds by \eqref{eq:mean_ibp}, the final equality holds by the (Lebesgue) fundamental theorem of calculus, and the inequality holds for any sufficiently small $\eps>0$ because $(z-y) / \eps \int_x^y \pi \to +\infty$ as $\eps \to 0$, whereas the left-hand side is bounded above since $\int_A (G_x-G_y)>0$. Hence by \hyperref[theorem:pratt]{Pratt's theorem}, $x' \mapsto \int v_\eps \dd G_{x'}$ is not less risk-averse than $v_\eps$---a contradiction.
	\end{proof}
	\renewcommand{\qedsymbol}{$\blacksquare$}

	Define $L : (-\infty,\sup X) \to [0,1]$ and $R : \bigcup_{x \in X} [x,+\infty) \to [0,1]$ by
	\begin{align*}
		L(y) &\coloneqq \inf_{z \in X} G_z(y)
		\quad \text{for each $y \in (-\infty,\sup X)$}
		\\
		\text{and} \quad
		R(y) &\coloneqq \sup_{z \in X} G_z(y) 
		\quad \text{for each $y \in \bigcup_{x \in X} [x,+\infty)$.}
	\end{align*}
	Evidently $L$ and $R$ are increasing. We claim that
	\begin{equation}
		G_x = L + \1_{[x,+\infty)}(R-L) 
		\quad \text{for every $x \in X$.}
		\label{eq:LR}
	\end{equation}
	To see why, fix any $x \in X$ and $y \in \R$. If $y<x$, then since $z \mapsto G_z$ is pointwise decreasing (i.e. increasing with respect to first-order stochastic dominance), there exists a sequence $(x^n)_{n \in \N}$ in $X \cap [x,+\infty)$ along which $G_{x^n}(y) \to L(y)$ as $n \to +\infty$, and $G_{x^n}(y) = G_x(y)$ for every $n \in \N$ by \hyperref[claim:LR]{claim~1}, so $G_x(y) = L(y)$. If $y \geq x$, then similarly, there exists a sequence $(x^n)_{n \in \N}$ in $X \cap (-\infty,x]$ satisfying $G_x(y) = G_{x^n}(y) \to R(y)$, so $G_x(y) = R(y)$.

	In case $\sup X < +\infty$, extend $L$ to $\R$ by letting $L \coloneqq \lim_{x \nearrow \sup X} L(x)$ on $[\sup X,+\infty)$. Note that $L$ is bounded. Hence letting $x \to \sup X$ in \eqref{eq:LR} yields that if $L \neq 0$ then $L/\sup L(\R)$ is a CDF concentrated on $X$.

	In case $\inf X > -\infty$, extend $R$ to $\R$ by letting $R \coloneqq 0$ on $(-\infty,\inf X)$ and, if $\inf X \notin X$, $R(\inf X) \coloneqq \lim_{x \searrow \inf X} R(x)$. Then setting $x = \inf X$ in \eqref{eq:LR} if $\inf X \in X$, and letting $x \to \inf X$ otherwise, we obtain that $R$ is an extended CDF concentrated on $X \cup \{\inf X\}$.

	\begin{namedthm}[Claim 2.]
		\label{claim:incr}
		$R-L$ is increasing.
	\end{namedthm} 

	\begin{proof}[Proof of {\hyperref[claim:incr]{claim~2}}]
		\renewcommand{\qedsymbol}{$\square$}
		Suppose toward a contradiction that $x'<y'$ and $R(x')-L(x') > R(y')-L(y')$ for some $x',y' \in \R$. 

		Since for each $z \in X$, $G_z$ is concentrated on $X$, we have $R \geq L$ with equality on $(-\infty,\inf X)$, and thus $x' \geq \inf X$. Since $L$ and $R$ are right-continuous, it follows that we may choose an $x'' \in (x',y')$ such that $R(x'')-L(x'') > R(y')-L(y')$ and $(-\infty,x''] \cap X \neq \varnothing$. Since $L/\sup L(\R)$ is a CDF concentrated on $X$ if $L \neq 0$, $R$ is an extended CDF concentrated on $X \cup \{\inf X\}$, and $y < \sup X$ (because $y \in X$ and $\sup X \notin X$), it follows that we may choose $x<y$ in $X$ such that $R(x)-L(x) > R(y)-L(y)$.%
			\footnote{Firstly, if $x'' \in X$, then let $x \coloneqq x''$, and otherwise, choose $x \in X \cap (-\infty,x'')$ such that $R(x)-L(x) > R(y')-L(y')$. Secondly, if $y' \in X$, then let $y \coloneqq y'$, and otherwise, choose $y \in (x,y')$ such that $R(x)-L(x) > R(y)-L(y)$.}

		The set $X \cap (y,+\infty)$ is non-empty since $y<\sup X$ (as $y \in X$ and $\sup X \notin X$). Fix an element $z \in X \cap (y,+\infty)$. Since $R$ and $L$ are right-continuous, we may choose an $\alpha \in (0,\min\{y-x,z-y\})$ such that
		\begin{equation}
			\inf_{x' \in [x,x+\alpha)} [R(x')-L(x')] > \sup_{x' \in [y,y+\alpha)} [R(x')-L(x')] .
			\label{eq:uniform_xy}
		\end{equation}
		Let $A \coloneqq (x,x+\alpha) \cup (y,y+\alpha)$.

		Fix a Lebesgue-integrable function $\pi : \R \to (0,+\infty)$. For each $\eps>0$, define $w_\eps : \R \to \R$ by $w_\eps(x') \coloneqq \int_0^{x'} (\1_A + \eps \pi \1_{\R \setminus A})$ for each $x' \in \R$, and note that $w_\eps$ is strictly increasing, bounded, and absolutely continuous, with (a.e.-defined) derivative $w_\eps'=1$ on $A$ and $w_\eps'=\eps \pi$ on $\R \setminus A$. For each $\eps>0$, let $v_\eps \coloneqq w_\eps|_X$ be the restriction of $w_\eps$ to $X$. For any sufficiently small $\eps>0$,
		\begin{multline*}
			\frac{\int v_\eps \dd (G_z-G_y)}{\int  v_\eps \dd (G_y-G_x)} 
			= \frac{\int_A (G_y-G_z) + \eps \int_{\R \setminus A} \pi (G_y-G_z)}{\int_A (G_x-G_y) + \eps \int_{\R \setminus A} \pi (G_x-G_y)}
			\\
			= \frac{\int_y^{y+\alpha} (R-L) + \eps \int_{y+\alpha}^z \pi (R-L)}{\int_x^{x+\alpha} (R-L) + \eps \int_{x+\alpha}^y \pi (R-L)} 
			< \frac{\alpha + \eps \int_{y+\alpha}^z \pi}{\alpha + \eps\int_{x+\alpha}^y \pi} = \frac{ v_\eps(z)- v_\eps(y)}{ v_\eps(y)- v_\eps(x)} ,
		\end{multline*}
		where the first equality holds by \eqref{eq:mean_ibp}, the second equality holds by \eqref{eq:LR}, the inequality holds for any sufficiently small $\eps>0$ by \eqref{eq:uniform_xy}, and the final equality holds by the (Lebesgue) fundamental theorem of calculus. Hence by \hyperref[theorem:pratt]{Pratt's theorem}, $x' \mapsto \int v_\eps \dd G_{x'}$ is not less risk-averse than $v_\eps$---a contradiction.
	\end{proof}
	\renewcommand{\qedsymbol}{$\blacksquare$}

	Define
	\begin{equation*}
		\lambda
		\coloneqq \sup_{ x \in (-\infty,\sup X) } [R(x)-L(x)] .
	\end{equation*}
	For any $x<y$ in $(\inf X,+\infty)$, by \eqref{eq:LR}, we have $R-L = G_x - G_y$ on $(x,y]$ and $G_x - G_y = 0$ on $\R \setminus (x,y]$, which since $G_x \geq G_y \neq G_x$ (as $z \mapsto G_z$ is strictly increasing with respect to first-order stochastic dominance) implies that there exists a $z \in (x,y]$ such that $R(z)>L(z)$. Since $R-L$ is increasing by \hyperref[claim:incr]{claim~2}, it follows that $R>L$ on $(\inf X,+\infty)$. Since $(-\infty,\sup X) \cap (\inf X,+\infty) \neq \varnothing$ (because $\sup X \notin X$), it follows that $\lambda > 0$.
	
	Define $G : \R \to \R_+$ by
	\begin{equation*}
		G(x) \coloneqq
		\begin{cases}
			[R(\inf X)-L(\inf X)] / \lambda
			& \text{if $x < \inf X$ or $x = \inf X \notin X$} \\
			[R(x)-L(x)] / \lambda
			& \text{if $x \in \co(X)$} \\
			1
			& \text{if $x \geq \sup X$.}
		\end{cases}
	\end{equation*}
	In case $\lambda=1$, let $H$ be an arbitrary CDF concentrated on $X$, and in case $\lambda<1$, define $H : \R \to \R_+$ by $H \coloneqq L/(1-\lambda)$ on $(-\infty,\sup X)$ and $H \coloneqq 1$ on $[\sup X,+\infty)$. Note that if $\lambda=1$, then $L=0$ by \hyperref[claim:incr]{claim~2}, since $L$ is increasing and $R \leq 1$. Hence whatever the value of $\lambda$, we have $G_x = \lambda G\1_{[x,+\infty)} + (1-\lambda)H$ for every $x \in X$ by \eqref{eq:LR}.

	$G$ is an extended CDF: it is increasing by \hyperref[claim:incr]{claim~2} and the definition of $\lambda$, right-continuous since $L$ and $R$ are, and satisfies $\lim_{x \nearrow +\infty} G(x) = 1$ by definition of $\lambda$. Furthermore, $G$ is concentrated on $X \cup \{-\infty\}$ since $R$ is and since if $L \neq 0$ then $L/\sup L(\R)$ is a CDF concentrated on $X$. Finally, $G = R-L > 0$ on $(\inf X,+\infty)$.

	It remains to show that if $\lambda<1$, then $H$ is a CDF concentrated on $X$. So assume that $\lambda<1$. Since both $L$ and $R-L$ are increasing (the latter by \hyperref[claim:incr]{claim~2}),
	\begin{equation*}
		\sup_{x \in (-\infty, \sup X)} H(x)
		= \sup_{x \in (-\infty, \sup X)} \frac{L(x)}{L(x) + 1 - R(x)}
		\leq 1 .
	\end{equation*}
	As $L$ is increasing, it follows that $H$ is increasing, and since $\lim_{x \searrow -\infty} L(x) = 0$ and $\lim_{x \nearrow +\infty} R(x) = 1$ (as $\lim_{x \searrow -\infty} G_y(x) = 0$ and $\lim_{x \nearrow +\infty} G_y(x) = 1$ for every $y \in X$) it follows that $\lim_{x \searrow -\infty} H(x) = 0$ and $\lim_{x \nearrow +\infty} H(x) = 1$. $H$ is right-continuous since $L$ is. This establishes that $H$ is a CDF. Finally, $H$ is concentrated on $X$ since if $L \neq 0$ then $L/\sup L(\R)$ is a CDF concentrated on $X$.
\end{proof}

\section{\texorpdfstring{\Cref{theorem:special}}{Theorem \ref{theorem:special}} with a greatest alternative}
\label{app:special_greatest}

When the hypothesis $\sup X \notin X$ is dropped from \Cref{theorem:special}, the result remains true as stated, except with \ref{item:special:oo} replaced by

\begin{enumerate}[label=(\alph*$^\star$)]

	\setcounter{enumi}{1}

	\item \label{item:special:oo_greatest}
	There exists an extended CDF $G$ concentrated on $X \cup \{-\infty\}$ with $G>0$ on $(\inf X,+\infty)$ such that for every bounded and strictly increasing function $v : X \to \R$, there exist $\alpha>0$ and $\beta \geq \beta^\star$ in $\R$ such that
	\begin{equation*}
		\alpha \int v \dd G_x + \beta
		= \int \max\{v(x),v(y)\} G(\dd y)
		\quad \text{for every $x \in X \setminus \{\sup X\}$}
	\end{equation*}
	and, if $\sup X \in X$, $\alpha \int v \dd G_{\sup X} + \beta^\star = \int \max\{v(\sup X),v(y)\} G(\dd y)$.
		
\end{enumerate}

To prove this, we establish a three-way equivalence between properties~\ref{item:special:lra}, \ref{item:special:oo_greatest}, and

\begin{enumerate}[label=(\alph*$^\star$)]

	\setcounter{enumi}{2}

	\item \label{item:special:distn_greatest}
	There exist $\lambda \in (0,1]$, an extended CDF $G$ concentrated on $X \cup \{-\infty\}$ with $G>0$ on $(\inf X,+\infty)$, and CDFs $H,H^\star$ concentrated on $X$ such that $H$ is first-order stochastically dominated by $H^\star$ and
	\begin{equation*}
		G_x =
		\begin{cases}
			\lambda G \1_{[x,+\infty)} + (1-\lambda)H
			& \text{if $x \in X \setminus \{\sup X\}$}
			\\
			\lambda G \1_{[x,+\infty)} + (1-\lambda)H^\star
			& \text{if $x = \sup X \in X$.}
		\end{cases}
	\end{equation*}
		
\end{enumerate}
The facts that \ref{item:special:distn_greatest} implies \ref{item:special:oo_greatest} and that \ref{item:special:oo_greatest} implies \ref{item:special:lra} follow from essentially the same (very short) arguments as in \cref{app:special_pf}.

\begin{proof}[Sketch proof that \ref{item:special:lra} implies \ref{item:special:distn_greatest}]
	Assume that $\sup X \in X$, and write $x \coloneqq \sup X$. The case $\abs*{X} \leq 2$ is trivial; assume for the remainder that $\abs*{X} \geq 3$. Proceed as in \cref{app:special_pf}, establishing \hyperref[claim:LR]{claim~1} on $X \setminus \{x\}$, defining $L$ and $R$ on $X \setminus \{x\}$ (rather than on $X$), and extending $L$ to $\R$ by insisting that $R-L$ (rather than $L$) be continuous at $x$. This argument delivers $\lambda \in (0,1]$, an extended CDF $G$ concentrated on $X \cup \{-\infty\}$, and a CDF $H$ concentrated on $X$ such that $G_z = \lambda G \1_{[z,+\infty)} + (1-\lambda)H$ for every $z \in X \setminus \{x\}$.

	It remains to show that there exists a CDF $H^\star$ concentrated on $X$ which first-order stochastically dominates $H$ and has $G_x = \lambda G \1_{[x,+\infty)} + (1-\lambda)H^\star$.

	\begin{namedthm}[Claim 3.]
		\label{claim:max}
		$G_x \leq L$ on $(-\infty,x)$.
	\end{namedthm}

	\begin{proof}[Proof of {\hyperref[claim:max]{claim~3}}]
		\renewcommand{\qedsymbol}{$\square$}
		Let $y = \sup (X \setminus \{x\})$, and note that $G_x \leq L$ on $(-\infty,y)$ by definition of $L$, since $x \mapsto G_x$ is increasing with respect to first-order stochastic dominance. If $y=x$, then there is nothing left to prove. Assume for the remainder that $y<x$. Since $G_x$ is concentrated on $X$, it is constant on $[y,x)$. It remains only to show that $G_x(y)\leq L(y)$.

		Fix an arbitrary $z \in X \setminus \{x,y\}$; this is possible since $\abs*{X} \geq 3$. Further fix a Lebesgue-integrable function $\pi : \R \to (0,+\infty)$. For each $\eps>0$, define $w_\eps : \R \to \R$ by $w_\eps(x') \coloneqq \int_0^{x'} ( \1_{[y,x]} + \eps \1_{[z,y]} + \eps \1_{\R \setminus [z,x]}\pi )$ for each $x' \in \R$, and note that $w_\eps$ is strictly increasing, bounded, and absolutely continuous, with (a.e.-defined) derivative $w_\eps' = \1_{[y,x]} + \eps \1_{[z,y]} + \eps \1_{\R \setminus [z,x]} \pi$. For each $\eps>0$, let $v_\eps \coloneqq w_\eps|_X$ be the restriction of $w_\eps$ to $X$. For any sufficiently small $\eps>0$, recalling that $R>L$ on $(\inf X,+\infty) \supseteq (z,y)$, we have
		\begin{multline*}
			\frac{ \int_y^x (G_y-G_x) + \eps\left(\int_z^y (G_y-G_x) + \int_{\R \setminus [z,x]}\pi(G_y-G_x) \right) }
			{ \eps \int_z^y (R-L) }
			\\
			\geq \frac{ \int w'_\eps (G_y-G_x) }
			{ \int w'_\eps (G_z-G_y) }
			= \frac{ \int v_\eps \dd (G_x-G_y) }
			{ \int v_\eps \dd(G_y-G_z) }
			\geq \frac{ v_\eps(x) - v_\eps(y) }
			{ v_\eps(y) - v_\eps(z) }
			= \frac{x-y}{\eps(y-z)} ,
		\end{multline*}
		where the first inequality follows from \eqref{eq:LR} since $G_y \leq G_z$ and $w'_\eps \geq 0$, the equality holds by \eqref{eq:mean_ibp}, and the second inequality holds by \hyperref[theorem:pratt]{Pratt's theorem} since $x' \mapsto \int v_\eps G_{x'}$ is less risk-averse than $v_\eps$ by the hypothesis \ref{item:special:lra}. Multiplying both sides by $\eps$ and letting $\eps \to 0$ yields 
		\begin{equation*}
			\frac{\int_y^x (G_y-G_x)}{\int_z^y (R-L)}
			\geq \frac{x-y}{y-z} .
		\end{equation*}
		Since $G_x$ and $G_y$ are constant on $[y,x)$, this is equivalent to
		\begin{equation*}
			G_y(y)-G_x(y)
			\geq \frac{1}{y-z}\int_z^y (R-L) .
		\end{equation*}
		
		Write $X^\dag \coloneqq X \setminus \{x,y\}$. By choosing $z \coloneqq \sup X^\dag$ if $\sup X^\dag \in X^\dag$, and letting $z \to \sup X^\dag$ in $X^\dag$ otherwise, we obtain $G_y(y)-G_x(y) \geq R(y)-L(y)$, since $R-L$ is continuous at $y$. Since $G_y(y) = R(y)$, by \eqref{eq:LR}, it follows that $G_x(y) \leq L(y)$, as desired.
	\end{proof}
	\renewcommand{\qedsymbol}{$\blacksquare$}

	Now, if $\lambda = 1$, then $L = 0$ (as noted in \cref{app:special_pf}), so $G_x = \1_{[x,+\infty)}$ by \hyperref[claim:max]{claim~3}. If instead $\lambda < 1$, then $G_x = \lambda\1_{[x,+\infty)} + (1-\lambda)H^\star$ where $H^\star \coloneqq (G_x - \lambda\1_{[x,+\infty)})/(1-\lambda)$. We have $H^\star \leq H \leq 1$ by \hyperref[claim:max]{claim~3} since $H = 1$ on $[x,+\infty)$. It follows (since $G_x$ is a CDF concentrated on $X$) that $H^\star$ is a CDF concentrated on $X$ which first-order stochastically dominates $H$.
\end{proof}

\end{appendices}



\printbibliography[heading=bibintoc]

\end{document}